\documentclass[12pt]{article}
\usepackage{amsmath}
\usepackage{graphicx}
\usepackage{natbib}
\usepackage[colorlinks,citecolor=black,urlcolor=black]{hyperref}
\usepackage{color}
\usepackage[dvipsnames]{xcolor}

\usepackage[capitalize,noabbrev]{cleveref}
\usepackage{amsthm,amsfonts}
\usepackage[boxruled,vlined]{algorithm2e}
\usepackage{array}
\usepackage[caption=false,font=normalsize,labelfont=sf,textfont=sf]{subfig}
\usepackage{textcomp}
\usepackage{stfloats}
\usepackage{verbatim}
\usepackage[english]{babel}

\newtheoremstyle{localrem}
	{5pt} 
	{5pt} 
	{\rm} 
	{} 
	{\bf} 
	{{\rm.}} 
	{.7em} 
	{} 

\theoremstyle{localrem}
\newtheorem{Definition}{Definition}[section]
\newtheorem{Remark}[Definition]{Remark}
\newtheorem{Example}[Definition]{Example}

\newtheoremstyle{localthm}
	{5pt} 
	{5pt} 
	{\sl} 
	{} 
	{\bf} 
	{{\rm.}} 
	{.7em} 
	{} 

\theoremstyle{localthm}
\newtheorem{Corollary}[Definition]{Corollary}
\newtheorem{Theorem}[Definition]{Theorem}
\newtheorem{Lemma}[Definition]{Lemma}

\newtheorem{Assumption}{Assumption}

\def\col{{:}}

\def\bs{\boldsymbol}

\def\F{\bs{F}}
\def\P{\bs{P}}

\def\QQ{\bs{Q}}
\def\G{\bs{G}}
\def\g{\bs{g}}
\def\k{\bs{k}}
\def\p{\bs{p}}

\def\x{\bs{x}}
\def\X{\bs{X}}
\def\y{\bs{y}}
\def\Y{\bs{Y}}
\def\z{\bs{z}}

\def\BB{{\cal B}}
\def\FF{{\cal F}}
\def\GG{{\cal G}}
\def\HH{{\cal H}}
\def\KK{{\cal K}}
\def\LL{{\cal L}}
\def\NN{{\cal N}}
\def\OO{{\cal O}}
\def\SS{{\cal S}}
\def\XX{{\cal X}}
\def\YY{{\cal Y}}

\def\N{\mathbb{N}}

\def\R{\mathbb{R}}
\def\Z{\mathbb{Z}}

\def\bpi{\bs{\pi}}
\def\bxi{\bs{\xi}}
\def\bi{\bs{i}}
\def\bone{\bs{1}}
\def\btwo{\bs{2}}

\def\ka{\kappa}

\def\hx{\hat{x}}
\def\hz{\hat{z}}

\def\diag{\mathrm{diag}}
\def\temp{\mathrm{temp}}
\def\old{\mathrm{old}}
\def\new{\mathrm{new}}

\def\QATS{\mathrm{QATS}}
\def\OS{\mathrm{OS}}
\def\OSH{\mathrm{OSH}}
\def\sOSH{\mathrm{sOSH}}
\def\ns{n_\mathrm{seeds}}

\def\hat{\widehat}

\def\Ex{\mathop{\rm I\!E}\nolimits}
\def\Pr{\mathop{\rm I\!P}\nolimits}

\def\argmax{\mathop{\rm arg\,max}}

\graphicspath{{./figs/}}
\begin{document}

\def\spacingset#1{\renewcommand{\baselinestretch}%
{#1}\small\normalsize} \spacingset{1}


  \title{\bf Quick Adaptive Ternary Segmentation: An Efficient Decoding Procedure For Hidden Markov Models}
  \author{Alexandre M\"osching\hspace{.2cm}\\
    Nonclinical Biostatistics, F.\ Hoffmann-La Roche, Switzerland\\
    Housen Li  and Axel Munk \\
    Institute for Mathematical Stochastics, Cluster of Excellence ``Multiscale \\Bioimaging: 
     from Molecular Machines to Networks of Excitable Cells'' \\
    Georg-August-Universit\"at G\"ottingen, Germany 
    }
  \maketitle

\bigskip
\begin{abstract}
Hidden Markov models (HMMs) are characterized by an unobservable Markov chain and an observable process---a noisy version of the hidden chain. Decoding the original signal from the noisy observations is one of the main goals in nearly all HMM based data analyses. Existing decoding algorithms such as Viterbi and the pointwise maximum a posteriori (PMAP) algorithm have computational complexity at best linear in the length of the observed sequence, and sub-quadratic in the size of the state space of the hidden chain.

We present Quick Adaptive Ternary Segmentation (QATS), a divide-and-conquer procedure with computational complexity polylogarithmic in the length of the sequence, and cubic in the size of the state space, hence particularly suited for large scale HMMs with relatively few states. It also suggests an effective way of data storage as specific cumulative sums. In essence, the estimated sequence of states sequentially maximizes local likelihood scores among all local paths with at most three segments, and is meanwhile admissible. The maximization is performed only approximately using an adaptive search procedure. Our simulations demonstrate the speedups offered by QATS in comparison to Viterbi and PMAP, along with a precision analysis. An implementation of QATS is in the R-package \texttt{QATS} on GitHub.
\end{abstract}

\noindent%

\section{Introduction}
\label{Sec:Intro}
%
A hidden Markov model (HMM), $(\X,\Y) = (X_k,Y_k)_{k\ge 1}$, defined on a probability space $(\Omega,\FF,\Pr)$, consists of an unobservable (hidden) Markov chain $\X$ on a finite state space $\XX = \{1,2,\ldots,m\}$, $m\ge 2$, and an observable stochastic process $\Y$ that takes values in a measurable space~$(\YY,\BB)$. 
Because of the generality of the observation and state spaces, HMMs are sufficiently generic to capture the complexity of various real-world time series, and meanwhile the simple Markovian dependence structure allows efficient computations. Upon more than half a century of development, HMMs and variants thereof have been established as one of the most successful statistical modeling ideas (see \citealp{Ephraim_2002}, \citealp{Cappe_2005} and \citealp{Mor_2021} for an overview). Since their early days, they are widely used in various fields of science and applications, such as speech recognition \citep{Rabiner_1989,GaYo08}, DNA or protein sequencing \citep{DEKM98,Kar09}, ion channel modeling \citep{BaRi92,PBSM21}, epidemiology \citep{TouEt20} and fluctuation characterization in macro economic time series \citep{Ham89,FS06}, to name only a few. 

Increasingly large and complex datasets with long time series have recently led to a revival in the development of scalable algorithms and methodologies for HMMs, see \cite{Bulla_2019} and the related literature in \cref{Sec:MAP,Sec:RiskBasedSeg}. In the present paper, we focus on computational aspects involved in the estimation of the hidden state sequence for large scale HMMs, that is when $n \gg m$. Here, $n$ denotes the length of the sequence of observations $\y=\y_{1:n} := (y_k)_{k=1}^n$ from $\Y_{1:n}:=(Y_k)_{k=1}^n$. The goal is, as \cite{Rabiner_1989} formulated it, to ``find the `correct' state sequence'' behind~$\y$. Procedures achieving such a task are commonly known as \textit{segmentation} or \textit{decoding} methods, and existing segmentation methods are tailored to what is exactly meant by the ``correct'' state sequence. 

{ We assume that, conditional on~$\X$, the components of the process~$\Y$ are stochastically independent, and each entry~$Y_k$ depends on~$\X$ only through the corresponding~$X_k$. This conditional independence structure is crucial for achieving computational efficiency, although it may be violated in practice. Such violations can often be mitigated, for example, by introducing additional hidden states or adopting a hierarchical modeling approach. Moreover,  we assume that the model parameters are (approximately) known.} These parameters are the transition matrix $\p^{(k)}:=\bigl(p_{ij}^{(k)}\bigr)\in[0,1]^{m\times m}$ of the Markov chain $\X$ with $p_{ij}^{(k)}=\Pr(X_{k+1}=j|X_k=i)$, the initial probability vector $\bpi:=(\pi_i)\in[0,1]^m$ with $\pi_i=\Pr(X_1=i)$, and the conditional distribution of~$Y_k$ given~$X_k=i$, which is assumed to have a density~$f_i^{(k)}$ with respect to some dominating measure~$\mu$ on~$\BB$. If the parameters are unknown, they are typically estimated via maximum likelihood or EM algorithms \citep{Baum_1966,Baum_1970,Baum_1972}. In case of temporal homogeneity, i.e., when $\p^{(k)}$ and $f_i^{(k)}$ do not depend on $k$, the estimation is usually not a severe burden as it can be sped-up for instance by using a fraction of the complete data \citep{GHS98}, while keeping statistical precision accurate, particularly for large scale data. 

For natural numbers $1\le \ell\le r$ and a vector $\bxi$ of dimension at least $r$, we write $\ell\col r$ for an index interval $\{\ell,\ell+1,\ldots,r\}$, and call it an \textit{interval}, and write $\bxi_{\ell:r}$ for the vector $(\xi_k)_{k=\ell}^r$. We call \textit{segments} of $\bxi$ the maximal intervals on which $\bxi$ is constant, i.e.\ $\ell\col r$ is a segment of $\bxi$ if there exists $\xi$ such that $\xi_k = \xi$ for all $k\in\ell\col r$, $\xi_{\ell-1}\neq \xi_\ell$ (if $\ell > 1$) and $\xi_r\neq \xi_{r+1}$ (if the dimension of $\bxi$ is strictly larger than $r$). We interpret vectors as row-vectors. {Superscripts denote dimensions if they are defined in the manuscript, otherwise powers.}
\subsection{Maximum a posteriori --- Viterbi path}
\label{Sec:MAP}
%
The most common segmentation method aims to find the most likely state sequence~$\x$ given observations~$\y$. It seeks a path~$\x\in \XX^{n}$ which is a mode of the complete likelihood
\begin{equation}
	\label{Eq:Pr_Xx_Yy}
	\Lambda_{1:n}(\x)
	\ := \
	\pi_{x_1} f_{x_1}^{{(1)}}(y_1) 
		\left(\prod_{k=2}^n p_{x_{k-1}x_k}^{{(k-1)}} f_{x_k}^{{(k)}}(y_k)\right).
\end{equation}
This sequence is commonly known as \textit{maximum a posteriori} (MAP) or \textit{Viterbi path}, named after \cite{Viterbi_1967} which determines such a path via dynamic programming, see \cite{Forney_1973} for details. In its most common implementation, Viterbi algorithm, simply referred to as \textit{Viterbi} in the sequel, has computational complexity~$\OO(m^2n)$, see also Algorithm~\ref{Alg:Viterbi}.

Due to the ever increasing size and complexity of datasets, there is interest in accelerating Viterbi \citep{Bulla_2019}. Several authors obtained sub-quadratic complexity in the size~$m$ of the state space. Specifically, \citet{Esposito_2009} modified Viterbi and achieved a best-case complexity of~$\OO(m\log(m)n)$. At each step of the dynamic program, their approach avoids inspecting all potential states by ranking them and stopping the search once a certain state is too unlikely. \cite{Kaji_2010} proposed to reduce the number of states examined by Viterbi by creating groups of states at each step of the dynamic program and iteratively modifying those groups, when necessary. In the best case, the complexity of their method is~$\OO(n)$. Both \cite{Esposito_2009} and \cite{Kaji_2010} have worst-case complexity equal to that of Viterbi. In contrast, \cite{Cairo_2016} were the first to achieve worst-case complexity~$\OO((m^2/\log m)n)$ (and an extra prepossessing cost that is polynomial in~$m$ and~$n$) by improving the matrix-vector multiplication performed at each step of the dynamic program. All those methods find a maximizer of \eqref{Eq:Pr_Xx_Yy}. Improving Viterbi by a polynomial factor in~$m$, or more, would have important implications in fundamental graph problems, as argued in \cite{Backurs_2017}.

There have been attempts at decreasing the computational complexity of Viterbi in the length~$n$ of the observed sequence~$\y$. \cite{Lifshits_2009} used compression and considered an observation space with finite support, i.e.~$\YY$ has finite cardinality. They proposed to precompress~$\y$ by exploiting repetitions in that sequence, and achieved varying speedups (e.g.\ by a factor~$\Theta(\log n)$) depending on the compression scheme. \cite{Hassan_2021} presented a framework for HMMs which allows to apply the parallel-scan algorithm \citep{LaFi80, Blelloch_1989} for parallel computation of the forward and backward loops of Viterbi. This parallel framework can achieve a span complexity of~$\OO(m^2\log n)$, but  necessitates a number of threads that is proportional to~$n/\log(n)$ and results in a total complexity of~$\OO(m^2 n)$. 

\subsection{Other risk-based segmentation methods}
\label{Sec:RiskBasedSeg}
%
Maximizing the complete likelihood as executed by Viterbi may share common disadvantages with other MAP estimators, see \cite{CaLa08}. For instance, Viterbi may perform unsatisfactorily if there are several concurring paths with similar probabilities. A different optimality criterion determines, at each time~$k\in 1\col n$, the most likely state~$\hat x_k$ which gave rise to observation~$y_k$, given the whole sequence~$\y$. The solution to this problem minimizes the expected number of misclassifications and is known as the \textit{pointwise maximum a posteriori} (PMAP) estimator, which is often referred to as \emph{posterior decoding} in bioinformatics and computational biology \citep{DEKM98}. To obtain the PMAP, a forward-backward algorithm similar to Viterbi computes the so-called \textit{smoothing} and \textit{filtering distributions}, which give the distribution of~$X_k$ given~$\Y_{1:n}$ and~$X_k$ given~$\Y_{1:k}$, $k\in 1\col n$, respectively, see \cite{Baum_1970} and \cite{Rabiner_1989}. The PMAP also has computational complexity~$\OO(m^2n)$. There is an important drawback of the PMAP paradigm for estimation: The resulting sequence~$\hat{\x}$ is potentially inadmissible, i.e.\ the probability to transition from~$\hx_k$ to~$\hx_{k+1}$ for some~$1\le k < n$ is zero.

\cite{Lember_2014} studied the MAP and PMAP in a risk-based framework, where both estimators are seen as minimizers of specific risks. By mixing those risks and other relevant ones, hybrid estimators combining desirable properties of both estimators are defined, see also \cite{Fariselli_2005}. In this case, suitable modifications of the forward-backward algorithm are possible to maintain a computational complexity of~$\OO(m^2n)$.

Provided that there is a priori knowledge about the number of segments of the hidden path, \cite{Titsias_2016} determined a most likely path with a user-specified number~$s$ of segments (sMAP). Precisely, they attempt to maximize~$\Lambda_{1:n}(\x)$ over all paths~$\x\in\XX^n$ such that the cardinality of~$\{k: x_k\neq x_{k+1}\}$ is equal to~$s-1$. The complexity of their method is~$\OO(s m^2n)$, and if one desires to look at all paths with up to~$s_{\max}$ segments, the overall complexity amounts to~$\OO(s_{\max}m^2n)$.

\subsection{Our contribution}
%
We present a novel decoding procedure---inspired by Viterbi and sMAP---achieving polylogarithmic computational complexity in terms of the sample size~$n$. Our method is particularly beneficial for HMMs with relatively infrequent changes of hidden states, since the case of frequent changes approaches a linear computational complexity, let alone to output the changes of state. The segmentation of HMMs with infrequent changes can be viewed as a particular problem of (sparse) change point detection, see recent surveys \citep{NHZ16,TOV20}. 

From this point of view, we introduce \textit{Quick Adaptive Ternary Segmentation (QATS)}, a fast segmentation method for HMMs. In brief, QATS sequentially partitions the interval~$1\col n$ into smaller intervals with the following property: On each interval, the state sequence that maximizes a localized version of the complete likelihood~\eqref{Eq:Pr_Xx_Yy}, over all state sequences with at most three segments (at most two changes of state), is in fact a constant state sequence (it has a single segment, i.e.\ no change of state). Thus, if at a certain stage of the procedure the maximizing sequence in a given interval was made of two or three segments (one or two changes of state), then those two or three segments would replace the original interval in the partition and those new segments would subsequently be investigated for further partitioning. 

This divide-and-conquer technique builds on the classical binary segmentation \citep{Bai_1997}, which allows at most one split at a time and is primarily used for breakpoint detection in economic time series. The idea of binary segmentation can be traced back to earlier work in cluster analysis \citep{ScKn74}. Here, we allow up to two splits per iteration (three new segments), granting it the name of \textit{ternary segmentation}. The benefit of considering three segments instead of two is the significant increase in detection power of change points, see Lemma~\ref{Lem:Charac_CP}. One could consider more than three segments for the sake of further improvement in detection power, but this would come at the cost of a heavier computational burden. A variant of ternary segmentation that searches for a bump in a time series was first considered in \cite{LeKl85}, the idea of which can also be found in circular binary segmentation \citep{OVLW04}. However, the concept of optimizing over two sample locations is much older and can already be found in the proposal by \cite{Page55}. 

To the best of our knowledge, binary or ternary segmentation, or any extension thereof, has not yet been used for decoding HMMs so far. A possible reason is that the resulting path, which we call \textit{QATS-path}, does not maximize an explicit score defined a priori, unlike the MAP, PMAP or sMAP discussed previously. Instead, the QATS-path solves a problem defined implicitly via recursive local maximizations. {This greedy nature hinders a thorough theoretical analysis on its statistical performance. However, in a simple scenario with~$m=2$, we are able to provide a mathematical justification for QATS. Our simulation study (including scenarios with~$m>2$) shows that QATS estimates the true hidden sequence with a precision comparable to that of its competitors, while being substantially faster already for moderate sized datasets. The empirically observed speedups are supported by our complexity analysis, which shows that QATS has computational complexity~$\OO(sm^3\log n)$, with~$s$ the number of segments of the QATS-path, for general scenarios with $m$ hidden states.} In case of a small number $m$ of states and a large number $n$ of observations, this can be significantly faster than the computational complexity $\OO\bigl(m^2 n/ \log(mn)\bigr)$ of the state-of-the-art accelerations of Viterbi, see \cref{Sec:MAP}. {This situation includes most applications of HMMs in electrophysiology \citep{venkataramanan2002applying} and in bioinformatics \citep{yoon2009hidden}, where typically $m = 2, 4$ or $20$.  See \cref{Sec:Data} for an illustrative example of  application.}

To achieve this important speedup, the maximization step at each iteration of the ternary segmentation is only performed approximately, in the sense that the best path with at most three segments may not be obtained, but a sufficiently good one will be found quickly. To rapidly obtain this path, we devise an adaptive search strategy inspired by the optimistic search algorithm of \cite{Kovacs_2020} in the context of change point detection. The original idea of adaptive searches can be traced back to golden section search \citep{Kie53}. The application of such an adaptive search is beneficial here because data are stored as cumulative sums of log-densities evaluated at the observations~$\y$.

The reasons why the surprising speed-up from~$\OO(n)$ to~$\OO(\log n)$ with little loss of statistical performance is possible at all can be summarized as follows:

\noindent 1) \emph{The switch of optimization perspectives.} {We search for sample locations at which the ``correct'' hidden path most likely switches its states, instead of finding the most likely hidden state for every sample location, like Viterbi and PMAP.}

\noindent 2) \emph{The search of three segments in each step.} {The choice of three segments considerably improves the statistical performance of using two segments, while introducing only a small computational cost, in particular, when the number of states is small. }

\noindent 3) \emph{The estimation of changes via local optima.} We demonstrate that the locations at which the hidden states change can be characterized through the likelihood score by local optima, which can be estimated much faster than a global one. It is the search for a local optimum rather than the global one that makes a fast algorithm requiring only~$\OO(\log n)$ evaluations of likelihood scores possible. 

\noindent 4) \emph{The use of a local likelihood score.} The local likelihood score has the benefit of being computable in~$\OO(1)$ operations since it consists in differences of certain partial sums under proper transformation.

The procedures devised in this article are collected in the R-package QATS and are available from \url{https://github.com/AlexandreMoesching/QATS}. All methods are also implemented in C++ using the linear algebra library Armadillo \citep{Sanderson_Curtin_2016, Sanderson_Curtin_2018} and made accessible to R \citep{R_2022} using Rcpp \citep{Rcpp_JSS_2011, Rcpp_Springer_2013, Rcpp_AS_2018} and RcppArmadillo \citep{RcppArmadillo_2014}.

The article is organized as follows. In \cref{Sec:Procedure}, we devise QATS, and provide computational guarantees and a computational complexity analysis of QATS. The theoretical results on QATS with methodological justifications, sensitivity analysis, and path properties are given in \cref{Sec:Analysis}. In \cref{Sec:MonteCarloSimulations}, we examine empirical performances of QATS in terms of estimation accuracy and computational efficiency. { Further, we demonstrate the utility of QATS using a real-world array CGH dataset in \cref{Sec:Data}.} Proofs, technical details and additional simulations are deferred to the Supplementary Material.

\section{Description of the procedure}
\label{Sec:Procedure}
%
The idea of QATS is to sequentially partition, or \textit{segment}, the interval~$1\col n$ into~$s\ge 1$ contiguous and sorted intervals~$S_1,S_2,\ldots,S_s$, that is~$S_u=\ell_u\col r_u$ with indices~$\ell_u\le r_u$, $u\in 1\col s$, such that~$\ell_1=1$, $r_s=n$ and~$r_u+1=\ell_{u+1}$, $u\in 1\col (s-1)$. The segmentation achieves the following goal: On each interval~$S_u$, the best local path with at most three segments is a constant path, i.e.\ it is made of only one segment. 

Precisely, a path of length $d$ with~$c\ge 1$ segments is a vector~$\x\in\XX^{d}$ with~$c-1$ breaks, or \textit{change points}: If~$c>1$, there exists~$\ka_0:=1 < \ka_1 < \cdots < \ka_{c-1} < \ka_c := d+1$ such that~$\x_{\ka_{u-1}\col (\ka_u-1)}$ is a constant vector and~$x_{\ka_{u-1}}\neq x_{\ka_u}$, for~$u\in 1\col c$. A constant path is thus the one that satisfies~$\#\{x_k:k\in 1\col d\}=c=1$, i.e.\ it is made of a single segment. Furthermore, on a given interval~$S=\ell\col r$, we define the \textit{local likelihood} of~$\x\in\XX^{r-\ell+1}$ and~$\y_{\ell\col r}$, given a previous state~$X_{\ell-1}=x_0\in\XX$, as the following quantity:
\begin{equation}
	\label{Eq:LocalLikelihood}
	\Lambda_{\ell:r}(\x| x_0)
	:=
	\begin{cases}
		\Lambda_{1:r}(\x) \text{ defined in~\eqref{Eq:Pr_Xx_Yy}}
			& \text{if}\ \ell = 1, \\
		\prod_{k=\ell}^r p_{x_{k-\ell},x_{k-\ell+1}}^{{(k-1)}} f_{x_{k-\ell+1}}^{{(k)}}(y_k) 
			& \text{otherwise}.
	\end{cases}
\end{equation}
If~$\ell=1$, the likelihood is independent of~$x_0$, so we either write~$\Lambda_{\ell:r}(\x)$ or let~$x_0$ be arbitrary.

Consequently, a \textit{best local path} on~$S_u=\ell_u\col r_u$ with at most three segments and previous state~$x_0\in\XX$ is a vector~$\x^*\in\XX^{r_u-\ell_u+1}$ which maximizes~$\Lambda_{\ell_u:r_u}(\x| x_0)$ over all vectors~$\x\in\XX^{r_u-\ell_u+1}$ with at most three segments. This section is devoted to the explicit construction of the procedure achieving the aforementioned segmentation.

\subsection{Ternary segmentation}
\label{Sec:TernarySegmentation}
%
The segmentation of~$1\col n$ is performed sequentially via \textit{ternary segmentation}. One operates with a tuple~$\SS=(S_u)_{u=1}^s$ of~$s$ contiguous and sorted intervals {(with $s$ being initially equal to $1$, and incrementing as the algorithm proceeds)}, a vector~$\hat{\z}\in\XX^s$ keeping track of the estimated state value on each interval, and a number~$u \in 1\col s$ denoting the current interval under investigation. At any stage~$u$ of the procedure, we may replace a single interval~$S_u$ by two or three new contiguous ones, as well as a scalar-state~$\hat z_u$ by a vector of two or three states. {As such, the size $s$ of $\SS$ and $\hat{\z}$ may be incremented by $1$ or $2$, respectively.}

The ternary segmentation proceeds as follows:
\begin{enumerate}
\item[(0)] Initially, set the current interval under investigation to be the whole interval~$1\col n$ and the tuple~$\SS$ to contain only that interval. The only estimated state is arbitrarily set to~$1$: $\SS \leftarrow (S_1) := (1\col n)$, $\hat{\z} \leftarrow 1$, $s \leftarrow 1$, $u \leftarrow 1$.
\item[(1)]
For the current interval~$S_u=\ell_u\col r_u$ of size~$d_u=r_u-\ell_u+1$, find the best local path~$\x^*\in\XX^{d_u}$ with at most three segments and previous state~$\hz_{u-1}$ (if~$u>1$). This yields a segmentation of~$S_u$ into~$\hat{c}\in 1\col 3$ contiguous interval(s)~$(S^w)_{w=1}^{\hat{c}}$ with state(s)~$(i^w)_{w=1}^{\hat{c}}$. Update~$\SS$, $\hat\z$ and $s$ accordingly:
\(
	S_u \leftarrow (S^w)_{w= 1}^{\hat{c}},
	\
	\hat z_u \leftarrow (i^w)_{w=1}^{\hat{c}},	
	\
	s \leftarrow s + \hat{c} - 1.
\)
\item[(2)] 
In case~$\hat{c}\in 2\col 3$, set the first of the newly created intervals as the new interval under investigation (i.e.\ $u$ remains unchanged), and go back to (1).

In case~$\hat{c}=1$, i.e.\ the old~$S_u$ remains unchanged by (1), move to the next available interval: $u \leftarrow u + 1$.

If~$u > s$, the algorithm terminates. Otherwise, go back to (1). 
\end{enumerate}

At the end of the procedure, the estimated path~$\hat{\x}$ from~$\SS$ and~$\hat{\z}$ is 
\begin{equation}
	\label{Eq:QATS_path}
	\hat{\x}
	:=
	(
		\hz_1\bone_{d_1},\;
		\hz_2 \bone_{d_2},\;
		\ldots,\;
		\hz_s \bone_{d_s}
	),
\end{equation}
where~$\bone_d$ for~$d\in\N$ is the~$d$-dimensional vector of ones and $d_u$ is the size of $S_u$, $u\in 1\col s$. \cref{Fig:Ternary_Segmentation} displays three possible stages of the procedure.

\begin{figure}[!t]
\centering
\includegraphics[width=0.36\textwidth] 
{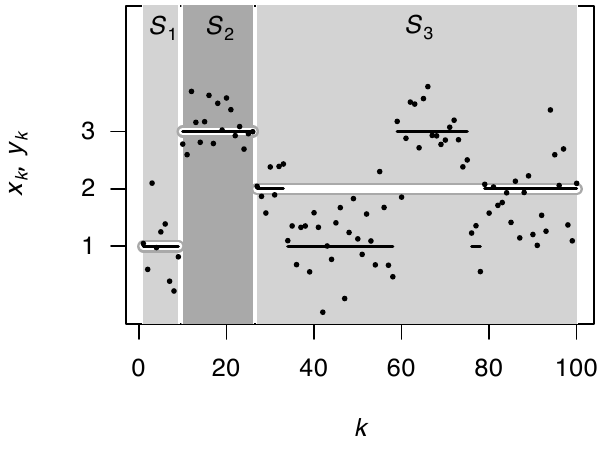}\,
\includegraphics[width=0.30\textwidth, trim = 48 0 0 0, clip]{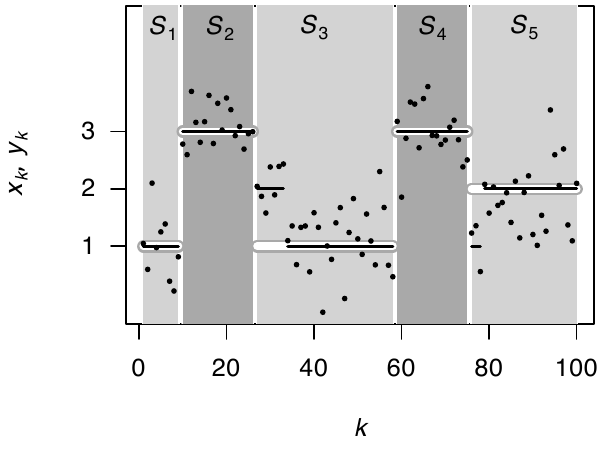}\,
\includegraphics[width=0.30\textwidth, trim = 48 0 0 0, clip]{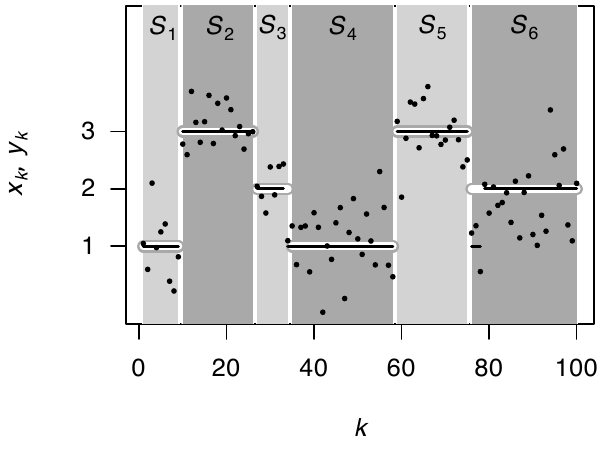}
\caption{Some stages of QATS. Black line segments represent the true, original path~$\x^o$, and points correspond to $\y$. \textit{Left:} After a few iterations of QATS, the best local path on~$S_1$ with at most three segments is constant equal to~$\hz_1 = 1$  (white segment with gray border). The best local path on~$S_2$ with at most three segments and previous states~$\hz_{1}=1$ is constant equal to~$\hz_{2}= 3$. The next interval to investigate is~$S_3$. \textit{Middle:} In step~(1), the search for the best local path on~$S_3$ with at most three segments yields a path with~$\hat{c}=3$ intervals (see \cref{Fig:Gain_Functions}). This replaces the old interval~$S_3$ by new ones: $S_3=S^{1}$, $S_4=S^{2}$, and $S_5=S^{3}$; and the scalar~$\hz_3=1$ by the vector of states: $\hz_{3:5}=(1,3,2)$. Step~(2) sets the new~$S_3$ as the next interval to investigate. \textit{Right:} Step~(1) replaced the old~$S_3$ by~$S_3=S^1$ and~$S_4=S^2$, and the old~$\hz_3$ by~$\hz_{3:4}=(2,1)$, whereas step~(2) sets the new~$S_3$ as the current interval. Applying steps~(1--2) to~$S_3$, followed by~$S_4$ and~$S_5$, will not yield any changes. The next interval to investigate will be~$S_6$.}
\label{Fig:Ternary_Segmentation}
\end{figure}

\subsection{Approximation}
%
To achieve sub-linear computational complexity in~$n$, the search of the best local path with at most three segments on a given interval~$S=\ell\col r$ of length~$d=r-\ell+1$ is only performed approximately, in the sense that the vector~$\x\in\XX^d$ with at most three segments achieving highest local likelihood~$\Lambda_{\ell:r}(\x| x_0)$ may not be found exactly. Instead, we perform an approximation by comparing the best constant path~$\x^{*1}$, the approximate best local paths (to be defined below) with two and three segments, respectively denoted~$\tilde{\x}^2$ and~$\tilde{\x}^3$, and selecting among those the path with the highest local likelihood score. Reasons for this approximation and elements of our procedure are detailed in this section.

The search of the best constant path on~$\ell\col r$ with previous state~$x_0\in\XX$ consists in the following maximization problem
\begin{equation}
	\label{Eq:H1}
	H^1 \ := \ \max_{i\in \XX}\, \Lambda_{\ell:r}(i \bone_{r-\ell+1}| x_0).
\end{equation}
The dependence on~$\ell$, $r$ and~$x_0$ for~$H^1$ is omitted to facilitate notation, and is therefore implicit. The same principle will be used in the sequel when convenient.

The best constant path is~$\x^{*1} := i^* \bone_{r-\ell+1}$ with $i^*
	:=
	\argmax_{i\in \XX}\, \Lambda_{\ell:r}(i \bone_{r-\ell+1}| x_0)$. 
This search costs~$m$ evaluations of~$\Lambda_{\ell:r}$ which, after preprocessing the data (see \cref{Sec:Linearization}), is a feasible task since it is independent of~$d$.

The search of the best path on~$\ell\col r$ with two or three segments is computationally more involved, since it requires the search of maxima of the following two target functionals
\begin{align}
	H^2(k)
	\ &:= \
	\max_{i_1\neq i_2}\,
	\Lambda_{\ell:r}\bigl((i_1 \bone_{k-\ell},i_2 \bone_{r-k+1})| x_0\bigr),
	\label{Eq:H2}\\
	H^3(\k)
	\ &:= \
	\max_{i_1\neq i_2\neq i_3}\,
	\Lambda_{\ell:r}\bigl((i_1 \bone_{k_1-\ell},i_2 \bone_{k_2-k_1},i_3 \bone_{r-k_2+1})| x_0\bigr),
	\label{Eq:H3}
\end{align}
over all~$k \in \KK^2_{\ell:r} := 
	\{k : \ell < k \le r\} $ and~$\k \in \KK^3_{\ell:r}:=
	\{(k_1,k_2) : \ell < k_1 < k_2 \le r\}$.
\cref{Fig:Gain_Functions} depicts the natural logarithm of the two maps~$H^2$ and~$H^3$ in the context of \cref{Fig:Ternary_Segmentation}.

\begin{figure}[!t]
\centering
\includegraphics[width=0.9\textwidth]{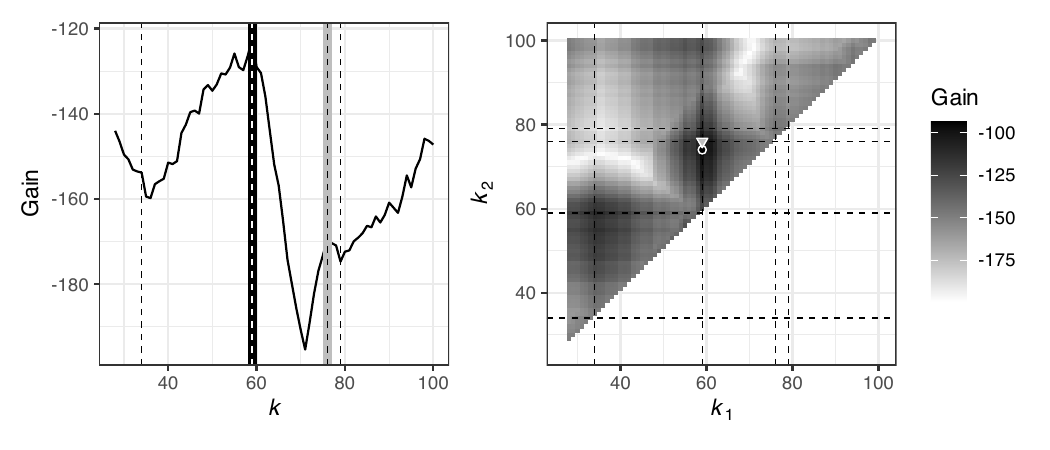}
\caption{Plots of~$\HH^2 = \log H^2$ (\textit{left}) and~$\HH^3 = \log H^3$ (\textit{right}) in the setting of the left plot of \cref{Fig:Ternary_Segmentation}, when the interval~$S_3$ is being investigated. Dashed lines (black, or white on black background) show the true change points of the hidden chain~$\x^o$, whereas the solid black line (\textit{left}) and point (\textit{right}) correspond to the respective global maxima of each map. The gray line (\textit{left}) and triangle (\textit{right}) correspond to the output of Algorithms~\ref{Alg:OSH2} and~\ref{Alg:OSH3}, respectively.} 
\label{Fig:Gain_Functions}
\end{figure}

Indeed, searching the global maximum of~$H^2$, respectively~$H^3$, would require~$(r-\ell)m(m-1)$, respectively~$(r-\ell)(r-\ell-1)m(m-1)^2/2$, probes of~$\Lambda_{\ell:r}$. When~$n$ and therefore~$\ell\col r$ are large, this task is computationally too costly, even after preprocessing the data (see \cref{Sec:Linearization}). Hence, we devise an approximate search algorithm to rapidly determine a one-dimensional local maximum of~$H^2$ and a two-dimensional local maximum~$H^3$, instead of their respective global maxima. Here, an index~$k^*\in\KK^2_{\ell:r}$ is called a \textit{one-dimensional local maximum} of~$H^2$ if~$H^2(k^*)\ge H^2(k)$ for~$k\in\KK^2_{\ell:r}$ such that~$|k^*-k|=1$. Likewise, a pair of indices~$\k^*\in\KK^3_{\ell:r}$ is called a \textit{two-dimensional local maximum} of~$H^3$ if~$H^3(\k^*)\ge H^3(\k)$ for~$\k\in\KK^2_{\ell:r}$ such that~$\|\k^*-\k\|_1=1$, where we define~$\|\bxi\|_1:=|\xi_1|+|\xi_2|$ for a vector~$\bxi = (\xi_1,\xi_2)\in\R^2$.

The approximate search of the best paths with two and three segments is inspired by an adaptive search algorithm known as \textit{optimistic search} (OS). OS was first used in detection of mean changes in independent Gaussian data by \cite{Kovacs_2020} to obtain sub-linear computational complexity when determining a new change point. In its simplest formulation, OS takes as an input a real-valued function~$H$ defined on an interval~$L\col R$, a tuning parameter~$\nu\in(0,1)$, and a maximal interval length~$d_o>1$, and returns a one-dimensional local maximum of~$H$ on~$L\col R$. The pseudocode for this procedure is given in Algorithm~\ref{Alg:OS}. Essentially, OS checks the value of $H$ at two points and keeps an interval that contains the point with the larger one. Hence, the kept interval contains at least one local maximum of~$H$ on~$\ell\col r$. The choice of two points ensures that a proportion of at least~$\nu/2$ points is excluded from the search interval. Thus, OS finds a local maximum of~$H$ on~$\ell \col r$ in~$\OO(\log (r- \ell))$ steps.

\begin{Lemma}[From \citealp{Kovacs_2020}]
\label{Lem:OS}
OS (as in Algorithm~\ref{Alg:OS}), returns a local maximum~$k^*\in L\col R$ of~$H$ in~$\OO(\log (R- L))$ steps/probes of~$H$, and its~$H$-value~$h^*=H(k^*)$, which is at least as large as any of the other probes performed during the algorithm.
\end{Lemma}

\begin{algorithm}[t]
\caption{Optimistic search:~$\OS(L, R, M, H, \nu, d_o)$}
\label{Alg:OS}
\vspace{1ex}
\KwIn{$L, R: L\le R$, $M\in (L\col R)\cup \{0\}$, $H: L\col R \to \R$, $\nu\in(0,1)$, $d_o>1$}
\KwOut{$(k^*, h^*)$, a local maximum of $H$ and its $H$-value}
\lIf{$M=0$}
    {$M \gets \lfloor (L+\nu R)/(1+\nu) \rfloor$}
\While{$R - L \ge d_o$}{
	\eIf{$R - M > M - L$}{
		$W \gets \lceil R - \nu (R - M) \rceil$\;
		\leIf{$H(W) > H(M)$}
		{$L \gets M$ ; $M \gets W$}
		{$R \gets W$}
		}{
		$W \gets \lceil L + \nu (M - L) \rceil$\;
		\leIf{$H(W) > H(M)$}
		{$R \gets M$ ; $M \gets W$}
		{$L \gets W$}
	}
}
$h^* \gets - \infty$\;
\For{$k \in L\col R$}{
	\lIf{$H(k) > h^*$}
	{$(k^*,h^*) \gets (k, H(k))$}
}
\end{algorithm}

In the sequel, we assume fixed values for the tuning parameter~$\nu$ and the minimal interval length~$d_o$, and therefore drop the dependence on those parameters when calling~$\OS$. Furthermore, since $H$ will be replaced by either $H^2$ or $H^3$, OS will return not only $k^*$ and $h^*$, but also the argument $\boldsymbol{i}^*$ that maximizes the corresponding local likelihood.

\subsubsection{One-dimensional OS on \texorpdfstring{$H^2$}{H2}}
%
To obtain a local maximum of~$H^2$ over~$\KK^2_{\ell:r}$, simply set~$H=H^2$, $L = \ell+1$, $R = r$ and apply OS. The resulting procedure, as shown in Algorithm~\ref{Alg:OSH2}, requires~$\OO(\log(r-\ell))$ iterations and returns~$k^*\in \KK^2_{\ell:r}$ such that $\tilde{\x}^2
	:=( i_1^* \bone_{k^*-\ell},i_2^* \bone_{r-k^*+1})
$
is an \textit{approximate best local path with two segments} on~$\ell\col r$ and previous state~$x_0$, with
$
	(i_1^*,i_2^*)
	:=
	\argmax_{i_1\neq i_2}\, 
	\Lambda_{\ell:r}\bigl((i_1 \bone_{k^*-\ell},i_2 \bone_{r-k^*+1})| x_0\bigr).
$

\begin{algorithm}
\caption{Optimistic search for~$H^2$: $\OSH^2(\ell,r,x_0)$}
\label{Alg:OSH2}
\vspace{1ex}
\KwIn{$\ell,r: r-\ell \ge 1$, $x_0\in\XX$}
\KwOut{$(k^*,h^*,\bi^*)$, a local maximum of~$H^2$, its~value and associated states}
$(k^*,h^*, \bi^*) \gets \OS \bigl(\ell+1, r, 0, H^2\bigr)$\;
\end{algorithm}

\subsubsection{Two-dimensional OS on \texorpdfstring{$H^3$}{H3}}
The procedure consists in an alternation of the fixed and varying arguments of~$H^3$ and the usage of Algorithm~\ref{Alg:OS} to the varying one.

\paragraph{Strategy}
The strategy for the two-dimensional OS can be broken down in three steps:

\noindent I. Initialization: Initialize~$h_\old$ and~$h_\new$ to~$-\infty$ and set some arbitrary index~$k_o\in (\ell+2)\col r$. We also set the initial solution~$\k^*$ to have~$k_o$ as a second component, i.e.\ $k_2^* = k_o$.

\noindent II. Horizontal search: For the first iteration or as long as~$h_\old$ is strictly smaller than~$h_\new$: Update~$h_\old \leftarrow h_\new$ and apply OS to the function~$H(k)=H^3(k,k_2^*)$ defined for~$k\in L\col R = (\ell+1)\col(k_2^*-1)$ using the current value of~$k_1^*$ as the first probe point (i.e.\ $M = k_1^*$ in the first step of OS), for all but the first iteration, for which the default initial probe is used. This so-called \textit{horizontal search} yields an index~$k^*$ which is a local maximum of~$H$ and which replaces the old value of~$k_1^*$. It also returns the score~$h_\new$ of that new~$\k^*$.
    
\noindent III. Vertical search: If~$h_\old$ is still strictly smaller than~$h_\new$ (which is necessarily the case for the first iteration), we perform a \textit{vertical search}: Update~$h_\old \leftarrow h_\new$ and apply OS to the function~$H(k)=H^3(k_1^*,k)$ defined for~$k\in L\col R = (k_1^*+1)\col r$ using the current value of~$k_2^*$ as the first probe point (i.e.\ $M = k_2^*$ in the first step of OS). This yields an index~$k^*$ which is a local maximum of~$H$ and which replaces the old value of~$k_2^*$, as well as the score~$h_\new$ of that new~$\k^*$.

Unless the new score is no larger than the old one at a certain stage, we alternate between the horizontal and vertical searches, swapping the roles of the fixed and varying components of~$H^3$ and applying OS to the varying one.

This alternating procedure strictly increases the score at each iteration, unless two consecutive ones yield the same score, in which case a local maximum of~$H^3$ is found. Indeed, because the current best index of the varying component of~$H^3$ is used as the first probe point~$M$ of OS, this ensures that any update of~$M$ in the while-loop of OS strictly increases the score. In contrast, if~$M$ remains unchanged in the while-loop and is returned at the end of the for-loop, then~$\k^*$ has not been updated twice in a row. Consequently,~$k_1^*$ is a ``vertical local maximum'' of~$H^3(\cdot,k_2^*)$ and~$k_2^*$ is a ``horizontal local maximum'' of~$H^3(k_1^*,\cdot)$. In other words,~$\k^*$ is a two-dimensional local maximum of~$H^3$.

\begin{Lemma}
\label{Lem:aOS}
Let~$V:\KK^3_{\ell:r} \to \R$ and suppose that every one-dimensional local maximum of~$V$ lies on an~$s\times s$ grid, i.e., there is a subset~$\KK$ of~$(\ell+1) \col r$ with cardinality~$s$ such that:

(i) For every~$k_1 \in (\ell+1) \col (r-1)$, all local maxima of~$V(k_1, \cdot)$ on~$(k_1+1)\col r$ are in~$\KK$;

(ii) For every~$k_2 \in (\ell+2) \col r$, all local maxima of~$V(\cdot, k_2)$ on~$(\ell+1) \col (k_2-1)$ are in~$\KK$.

\noindent Then, the alternation of horizontal and vertical searches returns a two-dimensional local maximum of~$V$ on~$\KK^3_{\ell:r}$ in~$\OO\bigl(s^2\log (r-\ell)\bigr)$ probes of~$V$.
\end{Lemma}

It will be shown (\cref{Sec:Justification}; cf.\ \cref{Fig:Gain_Functions_Population_Versions}) that the function~$H^3$ defined in \eqref{Eq:H3} fulfills the conditions of Lemma~\ref{Lem:aOS} in a noiseless scenario. The specific set~$\KK$ is shown to be~$\{\ka_1, \ldots, \ka_{c-1}\}$, where each~$\ka_a$ is a true change point. Thus, the alternating procedure returns a pair~$\k^*$ consisting of two true change points. An alternative way to treat vertical and horizontal searches which takes into account boundary effects of~$\KK^3_{\ell:r}$ is presented in Section~B of the Supplement.

\paragraph{Diagonal elements}
If the alternation of OS terminates at an element~$\k^*$ on the diagonal of~$\KK^3_{\ell:r}$, then that~$\k^*$ is in general a local maximum of~$H^3$. {Since, in this case, maximality is evaluated using at most two other elements $\k\in\KK^3_{\ell:r}$, we allow for an additional comparison with diagonal elements and proceed alternatively:} Apply OS to the function~$(\ell+1)\col (r-1)\ni k \mapsto H^3(k, k+1)$ with~$k_1^*$ as the first probe point, resulting in an element~$k^*$ and a (non-necessarily strict) increase of the score. Once~$\k^*$ has been updated to the new pair~$(k^*,k^*+1)$, the alternation between horizontal and vertical searches proceeds as explained earlier.

\paragraph{Maximum number of alternations}
To prevent the algorithm from performing too many alternations, we stop it if the number of iterations exceeds $v_o$. Our experiments show that~$v_o=20$ performs well. Should the algorithm terminate from this stopping criteria, the last update of~$\k^*$ and its corresponding~$H^3$-score are returned. The element~$\k^*$ then has no guarantee of being a local maximum of~$H^3$, but is necessarily the element with the largest $H^3$-score of all elements visited so far, including all the probes performed by OS.

\begin{algorithm}[t]
\caption{Seeded optimistic search for~$H^3$: $\sOSH^3(\ell,r,x_0,v_o,k_o)$}
\label{Alg:sOSH3}
\vspace{1ex}
\KwIn{$\ell,r: r-\ell \ge 2$, $x_0\in\XX$, $v_o>1$, $k_o\in (\ell+2)\col r$}
\KwOut{$(\k^*,h^*, \bi^*)$, an element of $\KK^3_{\ell:r}$, its $H^3$-value and associated states}
$\k^* \gets (\ell + 1, k_o)$; $h_\old \gets h_\new \gets -\infty$; $v \gets 1$; $\tau \gets 0$  ($0=\text{horizontal}$, $1=\text{vertical}$)\;
\While{{\normalfont ($h_\old < h_\new$ \textbf{and} $v < v_o$) \textbf{or} $v = 1$}}{
	$h_\old \gets h_\new$\;
	\eIf{$\tau=0$}{
		$(k_1^*, h^*, \bi^*) \gets \OS
			\bigl(\ell+1, k^*_2-1, k_1^*, H^3(\cdot,k^*_2)\bigr)$\;}{
		$(k_2^*,h^*, \bi^*) \gets \OS
			\bigl(k_1^*+1, r, k_2^*, H^3(k^*_1,\cdot)\bigr)$\;}
	\If{$k_1^*+1=k_2^*$}{
	    $(k^*, h^*, \bi^*) \gets \OS
	        \bigl(\ell + 1, r - 1, k_1^*, k \mapsto H^3(k, k+1)\bigr)$\; $\k^* \gets (k^*, k^*+1)$\;}
	$h_\new \gets h^*$; $v\gets v+1$; 
	$\tau \gets 1 - \tau$ (change direction)\;
}
\end{algorithm}

\paragraph{Complete two-dimensional search}
The procedure, as described to this point and which relies on an initial seed~$k_o\in (\ell+2)\col r$, is summarized in Algorithm~\ref{Alg:sOSH3}. Now we choose~$\ns\ge 1$ evenly spaced starting points~$k_o \in (\ell+2)\col r$, run Algorithm~\ref{Alg:sOSH3} for each of those \textit{seeds}, and select the endpoint~$\k^*$ with the largest~$H^3$-score. This method allows to increase the chance of finding an element~$\k^*\in\KK^3_{\ell:r}$ with a large value of~$H^3$. Simulations showed that~$\ns=3$ is a good trade-of between speed and exploration of the space~$\KK^3_{\ell:r}$.

The complete procedure is summarized in Algorithm~\ref{Alg:OSH3}. It returns~$\k^*\in \KK^3_{\ell:r}$ such that
$\tilde{\x}^3:=(i_1^* \bone_{k_1^*-\ell},i_2^* \bone_{k_2^*-k_1^*},i_3^* \bone_{r-k_2^*+1})$
is an \emph{approximate best local path with three segments} on~$\ell\col r$ and previous state $x_0$, where
\[
	(i_1^*,i_2^*,i_3^*) := 
	\argmax_{i_1\neq i_2\neq i_3}\, 
	\Lambda_{\ell:r}\bigl((i_1\bone_{k_1^*-\ell},i_2 \bone_{k_2^*-k_1^*},i_3 \bone_{r-k_2^*+1})| x_0\bigr).
\]

\begin{algorithm}[t]
\caption{Optimistic search for~$H^3$: $\OSH^3(\ell,r,x_0,v_o,\ns)$}
\label{Alg:OSH3}
\vspace{1ex}
\KwIn{$\ell,r: r-\ell \ge 2$, $x_0\in\XX$, $v_o>1$, $\ns \in 1\col (r-\ell-1)$}
\KwOut{$(\k^*,h^*, \bi^*)$, an element of $\KK^3_{\ell:r}$, its $H^3$-value and associated states}
$h^* \gets -\infty$ \;
\For{$i\in 1\col \ns$}{
	$k_o \gets \ell + 2 + \lfloor i \cdot (r-\ell-1)/(\ns+1) \rfloor$\;
	$(\k^\temp,h^\temp, \bi^\temp) \gets \sOSH^3(\ell,r,x_0,v_o, k_o)$\;
	\lIf{$h^\temp> h^*$}
	{$(\k^*,h^*, \bi^*) \gets (\k^\temp,h^\temp, \bi^\temp)$}
}
\end{algorithm}

\subsection{Linearization}
\label{Sec:Linearization}
%
The maximizations in \eqref{Eq:H1} to \eqref{Eq:H3} needed to compute~$H^c$, $c\in 1\col 3$, involve the evaluation of~$\Lambda_{\ell:r}$ at~$i \bone_{r-\ell+1}$, $(i_1 \bone_{k-\ell},i_2 \bone_{r-k+1})$ and $(i_1 \bone_{k_1-\ell},i_2 \bone_{k_2-k_1},i_3 \bone_{r-k_2+1})$. Consider for instance the computation of~$\Lambda_{\ell:r}\bigl((i_1 \bone_{k-\ell},i_2 \bone_{r-k+1})|x_0\bigr)$ when~$\ell>1$, which is {
\begin{multline*}
	\Lambda_{\ell:r}\bigl((i_1 \bone_{k-\ell},i_2 \bone_{r-k+1})|x_0\bigr)
	= \\
	p_{x_0i_1}^{(\ell-1)}f_{i_1}^{(\ell)}(y_{\ell})
	\left[
		 \prod_{t=\ell+1}^{k-1}p_{i_1i_1}^{(t-1)} f_{i_1}^{(t)}(y_t)
	\right]
	p_{i_1i_2}^{(k-1)}f_{i_2}^{(k)}(y_k)
	\left[
		 \prod_{t=k+1}^{r} p_{i_2i_2}^{(t-1)} f_{i_2}^{(t)}(y_t)
	\right].
\end{multline*}
Define
$
    \F
	=
	(F_{ik})
	:=
	\bigl(
		\prod_{t=1}^k (f_i^{(t)}(y_t)+1_{[f_i^{(t)}(y_t) = 0]})
	\bigr),
$ 
$
\bar{\F} = (\bar{F}_{ik}) := \bigl(\sum_{t = 1}^k 1_{[f_i^{(t)}(y_t) = 0]}\bigr),
$  
$\P = (P_{ik}):=
\bigl(\prod_{t=1}^k (p_{ii}^{(t)}+1_{[p_{ii}^{(t)} = 0]})\bigr)$ and 
$
\bar{\P} = (\bar{P}_{ik}):=
\bigl(\sum_{t=1}^k 1_{[p_{ii}^{(t)} = 0]}\bigr)
$ in $\mathbb{R}^{m\times n}$.
Then,}
\begin{multline*}{
    \Lambda_{\ell:r}\bigl((i_1 \bone_{k-\ell},i_2 \bone_{r-k+1})|x_0\bigr) 
	\ = \
	p_{x_0i_1}^{(\ell-1)}
	 p_{i_1i_2}^{(k-1)}
	 \frac{P_{i_1\, k-2}}{P_{i_1\, \ell-1}}
	 \frac{P_{i_2\,r-1}}{P_{i_2\,k-1}}
	\frac{F_{i_1\, k-1}}{F_{i_1\, \ell-1}}
	\frac{F_{i_2\, r}}{F_{i_2\,k-1}} }\\
	{\cdot 1_{[\bar{P}_{i_1\, \ell-1} =\bar{P}_{i_1\, k-2} ]}  1_{[ \bar{P}_{i_2\,k-1}=\bar{P}_{i_2\,r-1}]} 1_{[\bar{F}_{i_1\, \ell-1}= \bar{F}_{i_1\, k-1}]} 1_{[\bar{F}_{i_2\,k-1}=\bar{F}_{i_2\, r}]}.}
\end{multline*}{
Here the indicators prevent $F$ or $P$ from vanishing at index $(i, k)$ when $f_i^{(k)}(y_k) = 0$ or $p_{ii}^{(k)} = 0$, while $\bar{F}$ and $\bar{P}$ track such cases to ensure the likelihood is indeed zero, allowing for the general setting where $f_i^{(k)}(y_k)$ or $p_{ii}^{(k)}$ may be zero.}
Thus, precomputing the matrices~$\F$, { $\bar{\F}$, $\P$ and $\bar{\P}$} in~$\OO(mn)$ operations and memory allows to compute any~$\Lambda_{\ell:r}(i \bone_{r-\ell+1}|x_0)$, $\Lambda_{\ell:r}\bigl((i_1 \bone_{k-\ell},i_2 \bone_{r-k+1})|x_0\bigr)$ or $\Lambda_{\ell:r}\bigl((i_1 \bone_{k_1-\ell},i_2 \bone_{k_2-k_1},i_3 \bone_{r-k_2+1})|x_0\bigr)$ in just~$\OO(1)$ operations, that is, independently of the size of~$\ell\col r$, and therefore~$n$.

\begin{Lemma}
The computational costs of single evaluations of~$H^1$, $H^2(k)$ and $H^3(\k)$ are respectively~$\OO(m)$, $\OO(m^2)$ and $\OO(m^3)$.
\end{Lemma}

To prevent numerical instability due to multiplication and division of small quantities, we linearize all relevant expressions by applying the natural logarithm (see Section~\ref{App:LogLocLike} of the Supplement): We replace~$H^1$ by
\[
	\HH^1
	\ := \
	\max_{i\in \XX}\,
	\log \Lambda_{\ell:r} (i \bone_{r-\ell+1}|x_0 ),
\]
and, in Algorithms~\ref{Alg:OSH2} and \ref{Alg:OSH3}, we replace~$H^2$ and~$H^3$ by
\begin{align*}
    \HH^2(k)
	\ &:= \
	\max_{i_1\neq i_2}\,
	\log \Lambda_{\ell:r} \bigl((i_1 \bone_{k-\ell},i_2 \bone_{r-k+1})|x_0 \bigr), \\
	\HH^3(\k)
	\ &:= \
	\max_{i_1\neq i_2\neq i_3}\,
	\log \Lambda_{\ell:r} \bigl((i_1 \bone_{k_1-\ell},i_2 \bone_{k_2-k_1},i_3 \bone_{r-k_2+1})|x_0 \bigr),
\end{align*}
respectively. The above two maps are displayed in \cref{Fig:Gain_Functions}.

\begin{Corollary}
\label{Cor:Complexity_HHc}
The computational costs of~$\HH^1$, $\HH^2(k)$ and $\HH^3(\k)$ are respectively $\OO(m)$, $\OO(m^2)$ and $\OO(m^3)$.
\end{Corollary}

\subsection{Complete algorithm}
%
The pseudocode of the complete \textit{Quick Adaptive Ternary Segmentation} (QATS) algorithm from \cref{Sec:TernarySegmentation} is given in Algorithm~\ref{Alg:QATS}. It necessitates the preprocessing of data from \cref{Sec:Linearization} and uses Algorithms~\ref{Alg:OSH2} and \ref{Alg:OSH3} in which the~$H$-maps are replaced by their log-versions~$\HH$. {The hyperparameters are $\nu$, $d_o$, $v_o$, and $\ns$, and our experiments showed that tuning them has little impact on speed and precision.} The QATS-path~$\hat{\x}$ is then built from~$\SS$ and~$\hat \z$ as in~\eqref{Eq:QATS_path}.

\begin{algorithm}[t]
\caption{Quick~Adaptive~Ternary~Segmentation~($\QATS$)}
\label{Alg:QATS}
\vspace{1ex}
\KwIn{$\log \bpi$, {$\log\P$, $\log\F$, $\bar{\P}$, $\bar{\F}$,} $v_o>1$, $\ns \in 1\col (r-\ell-1)$}
\KwOut{$(\SS,\hat{\z})$, a segmentation of $1\col n$ and its state values}
$\SS \gets (1\col n)$; $\hat{\z} \gets 1$; $s \gets 1$; $u\gets 1$\;
\While{$u\le s$}{
	\lIf{$u> 1$}
	{$x_0\gets \hz_{u-1}$}
		$(\ell\col r) \gets S_u$; $(h^*_2,h^*_3) \gets (-\infty, -\infty)$; $h^*_1 \gets \HH^1$ with associate state $i_1^*$\;
	\If{$r-\ell \ge 1$}{
		$(k^*,h^*_2, \bi_2^*) \gets \OSH^2(\ell,r,x_0)$\;
		\If{$r-\ell \ge 2$}{
			$(\k^*,h^*_3, \bi_3^*) \gets \OSH^3(\ell,r,x_0,v_o,\ns)$\;}
	}
	$\hat{c} \gets \argmax_c\, h^*_c$\;
	\lIf{$\hat{c} = 1$}
	{$\hz_u \gets i_1^*$; $u \gets u+1$}
	\lIf{$\hat{c} = 2$}
	{$S_u \gets (\ell\col (k^* -1), k^*\col r)$ ; $\hz_u \gets \bi_2^*$; $s \gets s + 1$}
	\lIf{$\hat{c} = 3$}
	{$S_u \gets (\ell\col (k^*_1 -1), k^*_1\col (k^*_2-1), k^*_2\col r)$; $\hz_u \gets \bi_3^*$ ; $s \gets s + 2$}	
}
\end{algorithm}

\subsection{Computational complexity}
%
The Lemma below provides a bound on the number of iterations of QATS as displayed in Algorithm~\ref{Alg:QATS}. It involves the number~$s$ of intervals in~$\SS$ returned by QATS.

\begin{Lemma}
The number of iterations in the while loop of QATS is at most~$2s-1$.
\end{Lemma}

\begin{proof}
In the worst case, each of the~$s-1$ separations of~$1\col n$ is due to single splits ($\hat c = 2$), and no double splits ($\hat c = 3$). To confirm that on each interval of~$\SS$ the best local path with at most three segments is constant, one extra iteration ($\hat c = 1$) per interval is~necessary.
\end{proof}

Consequently, the number of calls of~$\OSH^2$ and the number of calls of~$\OSH^3$ are both bounded by~$2s-1$. Lemma~\ref{Lem:OS} implies that the number of probes of~$\HH^2$ in~$\OSH^2$ for a given interval~$\ell\col r$ is of the order~$\OO(\log(r-\ell))$. As to the complexity of~$\OSH^3$, it is broken down as follows: First, for each call of~$\OSH^3$, there are~$\ns$ calls of~$\sOSH^3$. Second, the number of alternations of vertical and horizontal searches in~$\sOSH^3$ for a given seed, and therefore the number of calls of~$\OS$, is bounded by~$v_o$, since Algorithm~\ref{Alg:sOSH3} stops as soon as the number~$v$ of alternations exceeds~$v_o$. Then again, for each call of~$\OS$, the number of probes of~$\HH^3$ is of the order~$\OO(\log(r-\ell))$. Finally, Corollary~\ref{Cor:Complexity_HHc} ensures that each query~$\HH^c$ costs~$\OO(m^c)$ operations,~$c\in 1\col 3$. This reasoning proves the following result:

\begin{Theorem}
QATS has computational complexity~$\OO(sm^3\log n)$.
\end{Theorem}

First, QATS performs particularly well with only a few segments in contrast to dynamic programming methods, which work independently of the number of segments. Second, the benefits of QATS are more pronounced for small values of~$m$ ($m \ll n/s$). Third, the process of storing the data has computational complexity~$\OO(mn)$ and is easily parallelized, e.g.\ via the parallel-scan algorithm \citep{LaFi80}, resulting in~$O(\log n)$ span complexity. That means, the overall complexity (without parallelization and proper storing of the data) is~$\OO\bigl(\max(mn,sm^3\log n)\bigr)$, with $\OO(mn)$  the computation of the~inputs~{$\log\P$, $\log\F$, $\bar{\P}$, $\bar{\F}$}~from~the~raw~data. { In practice (cf.\ \Cref{Sec:MonteCarloSimulations}), the number $s$ of segments in a QATS-path is often close to the expected number of segments in the underlying hidden Markov chain. This expectation depends explicitly on the transition matrices and initial distributions (see Lemma~\ref{Lem:ExpNumSeg}).}

\section{Theoretical analysis}
\label{Sec:Analysis}
%
In this section, we investigate theoretical properties of QATS, and, for technical simplicity, we focus on HMMs with two hidden states (formalized in Assumption~\ref{Aspt:SimpleSetting} below). 

\subsection{Justification}
\label{Sec:Justification}
%
At each step of QATS, a given interval may be split in two or three new intervals with the change point(s) being determined by the search of local maxima on~$\HH^2$ and~$\HH^3$. To justify this procedure, we show that there is indeed a one-to-one correspondence between local maxima of~$\HH^2$ and~$\HH^3$, and change points of the hidden signal. In the sequel, we study without loss of generality the case of~$\ell=1$ and~$r=n$, and therefore drop any dependence on~$\ell$ and~$r$ in the notation. 

Let~$\x^o$ be the true signal at the origin of the observations~$\y$. Let~$\KK^o:=\{k\in 2\col n: x_{k-1}\neq x_k\}$ be the set of true change points of~$\x^o$. That means, $\x^o$ consists of~$s^o=\#\KK^o+1$ segments. If~$s^o > 1$, the elements of~$\KK^o$ are written~$\ka_1< \cdots<\ka_{s^o-1}$. For convenience, we also define~$\ka_0:=1$ and~$\ka_{s^o}:=n+1$. Finally, we set a basic setting for which a mathematical analysis of~$\HH^2$ and~$\HH^3$ is tractable:
\begin{Assumption}
\label{Aspt:SimpleSetting}
Let $n\ge 3$, $m=2$, $\bpi = (1/2, 1/2)$ and $p_{12}^{{(k)}}=p_{21}^{{(k)}}=\varepsilon$ for some~$\varepsilon\in (0,1/2)$ { and all $k\in1\col n$}. Further, $\YY=\R$ and there exist constants~$\beta_1,\beta_2\in\R$, $\beta_2>0$, such that:
\begin{equation*}
	\{y_k:k\in 1\col n\} \  \in \ 1\col 2, \quad
	 \log f_i^{{(k)}}(y) \  = \ \beta_1 + \beta_2 1_{[y = i]},
	\qquad i,y\in 1\col 2, \, k\in 1\col n.
\end{equation*}
\end{Assumption}
{ The setting of Assumption~\ref{Aspt:SimpleSetting} is (temporally) homogeneous, i.e., $\p^{(k)} = \p$ and $f^{(k)}_i = f_i$, for all $k\in 1\col n$, and describes} the ideal case where the observation sequence~$\y$ completely characterizes the true signal~$\x^o$. Further, the probability that the Markov chain stays at a certain state is higher than that of jumping to another state, since~$p_{11}^{{(k)}}=p_{22}^{{(k)}}=1-\varepsilon > 1/2$.

\begin{Example}
Assumption~\ref{Aspt:SimpleSetting} holds, for instance, when~$\LL(Y_k|X_k=i)=\NN(i,\sigma^2)$,~$i\in 1\col 2$, for some~$\sigma>0$, and the dominating measure $\mu$ is the Lebesgue measure, since then~$\beta_1=-(\log (2\pi \sigma^2) + \sigma^{-2})/2$ and~$\beta_2= \sigma^{-2}/2$ satisfy the required property.
\end{Example}

In the next Theorem, we show that local maxima of~$\HH^2$ and~$\HH^3$ in the interior of~$\KK^2$ and~$\KK^3$ provide information on change points. Reciprocally, all change points appear either as local maxima of~$\HH^2$ or as components of local maxima of~$\HH^3$. In other words, the study of the maps~$\HH^2$ and~$\HH^3$ shall theoretically unveil all change points of the true sequence~$\x^o$.

\begin{Theorem}
\label{Lem:Charac_CP}
Suppose that~$s^o\ge 2$ and that Assumption~\ref{Aspt:SimpleSetting} holds. Then:

(i) Local maxima of~$\HH^2$ are located either at $2$, $n$ or $\ka_a$, $a\in 1\col (s^o-1)$. Local maxima of~$\HH^3$ are located either at boundary points~$(2,n)$,~$(2,\ka_a)$, $(\ka_a,n)$, $a\in 1\col (s^o-1)$, on the diagonal~$\{\k\in\KK^3:k_2=k_1+1\}$, or at pairs~$(\ka_a,\ka_b)$ with~$a,b\in 1\col (s^o-1)$ such that~$a<b$ and~$a+b$ is odd.

(ii) Conversely, for all~$a\in 1\col (s^o-1)$, either~$(\ka_{a-1}, \ka_a)$ or~$(\ka_a, \ka_{a+1})$ is a local maximum of~$\HH^3$. (Note that~$(\ka_0, \ka_1)$ or~$(\ka_{s^o -1}, \ka_{s^o})$ is a local maximum of~$\HH^3$ if and only if~$\ka_1$ or~$\ka_{s^o -1}$ is a local maximum of~$\HH^2$, respectively.)
\end{Theorem}

\cref{Fig:Gain_Functions_Population_Versions} exemplifies the above result. Part~(i) of Lemma~\ref{Lem:Charac_CP} indicates that local maxima found by QATS serve as proper estimates of change points (Lemmas~\ref{Lem:OS} and \ref{Lem:aOS}). Conversely, as shown in the left panel of \cref{Fig:Gain_Functions_Population_Versions}, not every change point corresponds to a local maxima of~$\HH^2$. However, by Part~(ii) of Lemma~\ref{Lem:Charac_CP}, the collection of local maxima of~$\HH^3$ allows us to find all change points. We stress that this is one of the motivations for QATS, which uses three segments, instead of two as it is usually done in change point detection.

\begin{figure}[!t]
\centering
\includegraphics[width=0.9\textwidth]{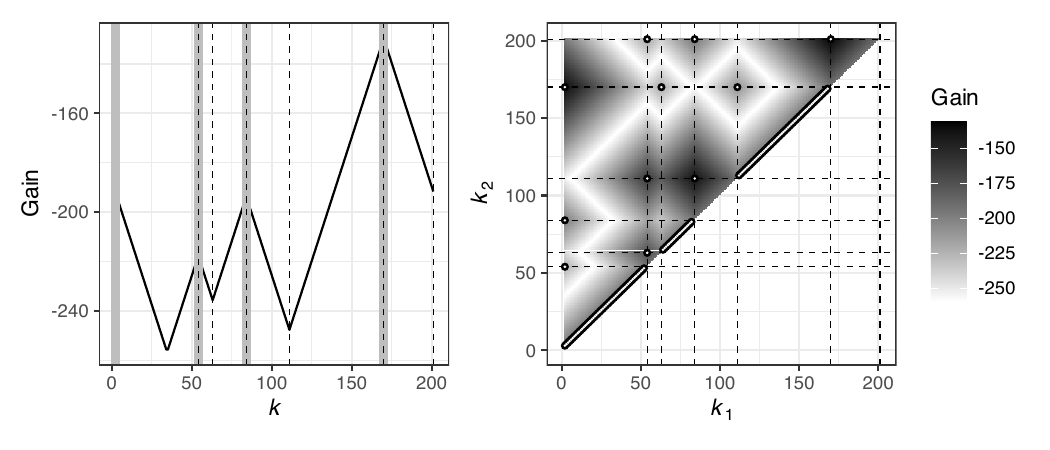}
\caption{Plots of~$\HH^2 = \log H^2$ (\textit{left}) and~$\HH^3 = \log H^3$ (\textit{right}). Local maxima (solid gray lines or gray points with black border) are located either on the boundaries, on the diagonal (for~$\HH^3$), on true change points or pairs thereof (dashed lines).} 
\label{Fig:Gain_Functions_Population_Versions}
\end{figure}

\subsection{Sensitivity}
\label{Sec:Sensitivity}
%
At each iteration of the while loop of Algorithm~\ref{Alg:QATS}, a local maximum~$k^*$ of~$\HH^2$ may be added to the list of change points if the local likelihood score of a path with two segments and jump at~$k^*$ is better than that of a constant path. Likewise, a local maximum~$\k^*$ of~$\HH^3$ may be added to the list of change points if the score of a path with three segments and jumps at~$\k^*$ is better than that of a constant path. The following results derive conditions on the relative position of change points in order for corresponding paths with two or three segments to have a higher score than a constant path. It also shows that if there is no change point ($s^o=1$), then the constant path indeed has a higher score than any other path with two or three segments. {These findings shed light on the detection power and estimation efficiency of QATS, which will facilitate future research on establishing statistical guarantees (e.g.\ risk bounds).} Define
$
	\delta 
	=
	\delta(\varepsilon,\beta_2)
	:=
	\beta_2^{-1} \log ((1-\varepsilon)/\varepsilon)
	>
	0.
$

\begin{Lemma}
\label{Lem:OptimH0H1H2}
Suppose that Assumption~\ref{Aspt:SimpleSetting} holds true. If~$s^o=1$, we then have that
$
    \HH^1 > \max_{k\in\KK^2}\, \HH^2(k)>\max_{\k\in\KK^3}\, \HH^3(\k).
$
If~$s^o=2$, we have
$
    \HH^2(\ka_1) > \max_{\k \in \KK^3}\, \HH^3(\k)
$
and the equivalence
$
	\HH^2(\ka_1) 
	> \HH^1
$
if, and only if,
$
	\ka_1 \in \bigl( 1+\delta , n + 1 - \delta \bigr).
$
If~$s^o=3$, we have that 
\begin{itemize}
\item  $\HH^2(\ka_1) > \HH^1$ if, and only if, $\ka_2 > \frac{n}{2} + 1 + \delta$ and $\ka_1 \in \bigl(1+\delta, 2\ka_2 - n - 1 - \delta \bigr)$;
\item  $\HH^2(\ka_2) > \HH^1$ if, and only if,  $\ka_1 < \frac{n}{2} + 1 - \delta$ and $\ka_2 \in \bigl(2\ka_1 - 1 + \delta, n + 1 - \delta \bigr)$;
\item $\HH^3(\ka_1,\ka_2) > \HH^1$ if, and only if,  $\ka_2 - \ka_1 \in ( 2\delta , n - 2\delta )$;
\item $\HH^3(\ka_1,\ka_2) > \max_{k \in \KK^2}\, \HH^2(k)$ if, and only if, $\ka_1>1+\delta$, $\ka_2<n+1-\delta$, $\ka_2- \ka_1 \ge \delta-1$.
\end{itemize}
\end{Lemma}

\begin{Lemma}
\label{Lem:H1>H0}
Suppose that Assumption~\ref{Aspt:SimpleSetting} holds and let~$s^o\ge 3$ and~$a\in 1\col (s^o-1)$. Then~$\HH^2(\ka_a) > \HH^1$ if, and only if,
\begin{equation*}
	\frac{\|\y_{1:(\ka_a-1)} - i_1 \bone\|_1}
	{\ka_a-1} 
	 < 
	\frac{1}{2} - 
	\frac{1}{2}
	\frac{\delta}{\ka_a-1}, \quad
	\frac{\|\y_{\ka_a:n} - i_2 \bone\|_1}{n-\ka_a+1}
	 < 
	\frac{1}{2} - 
	\frac{1}{2}
	\frac{\delta}{n-\ka_a+1},
\end{equation*}
for some~$(i_1,i_2)\in\{(1,2),(2,1)\}$.
\end{Lemma}

\begin{Lemma}
\label{Lem:H2>H0}
Suppose that Assumption~\ref{Aspt:SimpleSetting} holds, and let~$s^o\ge 3$ and~$a,b\in 1\col (s^o-1)$ be such that~$a<b$ and~$a+b$ is odd. Then~$\HH^3(\ka_a,\ka_b) > \HH^1$ if, and only if,
\begin{equation*}
	\frac{\|\y_{1:(\ka_a-1)} - i_1 \bone\|_1 + 
	\|\y_{\ka_b:n} - i_1 \bone\|_1}
	{n-(\ka_b-\ka_a)}
	 < \frac{1}{2} -
	\frac{\delta}{n-(\ka_b-\ka_a)}, \;
	\frac{\|\y_{\ka_a:(\ka_b-1)} - i_2 \bone\|_1}
	{\ka_b-\ka_a}
	 < \frac{1}{2} -
	\frac{\delta}{\ka_b-\ka_a},
\end{equation*}
for some~$(i_1,i_2)\in\{(1,2),(2,1)\}$.
\end{Lemma}

\subsection{Properties of QATS-paths}
%
The path~$\hat{\x}$ returned by Algorithm~\ref{Alg:QATS} is admissible, in the sense that all transitions in~$\hat{\x}$ have positive probability. Indeed, $\hat{\x}$ is built from the left to the right, taking into account the last state~$x_0$ from the previous segment for the derivations.

The QATS-path is different from paths resulting from risk-based segmentation techniques MAP, PMAP and sMAP in \cref{Sec:Intro}, as it does not maximize a risk function. Instead, it may be seen as a form of ``greedy'' decoder, since it selects a window of observation~$\ell\col r$ and determines the best path with at most three segments in that window.

Thus, if one omits the approximation due to OS in Algorithm~\ref{Alg:QATS}, the resulting path is an element of the following set
\begin{multline*}
    	\hat{\XX}
	:=	
	\bigl\{
		\x\in\XX^n :
		\Lambda_{\ell_u:r_u}(\x_{\ell_u:r_u}| x_{\ell_u-1}) 
		\ge
		\Lambda_{\ell_u:r_u}(\x'| x_{\ell_u-1}),
		\forall 
		\x'\in\bar{\XX}^{d_u} 
		\text{ and }
            \forall
		u \in 1\col s
	\bigr\},
\end{multline*}
where~$\bar{\XX}^d$ is the set of vectors of size~$d$ with at most three segments, and~$\ell_u$, $r_u$ and $d_u$ depend on the segmentation of the concerned~$\x$. In particular, $\hat{\XX}$ does not contain vectors~$\x$ for which a certain interval~$\ell_u\col r_u$ could be further split in two or three intervals and at the same time yield a higher local likelihood score. The next simple example shows that in general this set may contain more than one element.

\begin{Example}
Let $n=4$, $m=2$ and $\LL(Y_k|X_k=i)=\NN(i,4)$. We further set $\log \bpi = \theta_{\bpi} \bone$ and $p_{11} = p_{22} = \theta_{\p} + \log 2$, with constants~$\theta_{\bpi}=-\log 2$ and~$\theta_{\p}=-\log 3$. Suppose now that we observe~$\y=(1,4,-1,1)$. That means, with~$\theta_{f}=-\log(8\pi)/2$, we have
\[
	\bigl(\log f_i(y_k)\bigr)_{ik}
	\ = \
	\theta_{f} - \frac{1}{8}
	\begin{pmatrix}
	0 & 9 & 4 & 0 \\
	1 & 4 & 9 & 1
	\end{pmatrix}.	 
\]
Simple computations show that the best paths on~$1\col n$ with one, two or three segments are the following paths, respectively, along with their corresponding log-local likelihood scores: $\x^{(1)} := (1,1,1,1)$, $\log \Lambda_{1:n}(\x^{(1)}) = \theta - 13/8 + 3 \log 2$, $\x^{(2)} := (2,2,1,1)$, $\log \Lambda_{1:n}(\x^{(2)}) = \theta - 9/8 + 2 \log 2$, $\x^{(3)} := (1,2,1,1)$, $\log \Lambda_{1:n}(\x^{(3)}) = \theta - 8/8 + \log 2$, where~$\theta=\theta_{\bpi} + 3 \theta_{\p} + 4  \theta_{f}$. But since~$\log 2 \approx 0.69 > 4/8$ and~$2\log 2 \approx 1.39 > 5/8$, we find that~$\x^{(1)}$ is the best path with at most three segments, thus implying~$(1,1,1,1) \in \hat{\XX}$. To show that~$(2,2,1,1)\in\hat{\XX}$, too, one first verifies that, on~$1\col 2$, the best path out of the four possible ones with at most three (i.e.\ two) segments is the path~$(2,2)$, with a log-score of~$\theta_{\bpi} + \theta_{\p} + 2  \theta_{f} - 5/8 + \log 2$. Finally, out of the four possible path on~$3\col 4$ with previous state~$2$, the path~$(1,1)$ is the best one, with a score of~$2\theta_{\p} + 2  \theta_{f} - 4/8 + \log 2$. An exhaustive search shows that there are no other paths in~$\hat{\XX}$. 
\end{Example}

\section{Monte--Carlo simulations}
\label{Sec:MonteCarloSimulations}
%
The goals of the simulations are to:

\noindent 1) Verify empirically that QATS is substantially faster than Viterbi and PMAP when the number of expected segments is small compared to the length of the observation sequence.

\noindent 2) Show that the accuracy of QATS-paths is comparable to Viterbi and PMAP-paths.

\subsection{Simulation settings}
\label{Sec:Settings}
%
We consider sequences of observations~$\y = \y_{1\col n}$ of length~$n$ where $n \in 1 + \{10^3, 10^4, 10^5, 10^6\}$ (the additional~$1$ will be clear soon). For the size~$m$ of the state space, we study~$m \ \in \ \{2,3,5,10\}$. Much larger state spaces are not recommended for QATS because of its cubic complexity in the number of states. The initial distribution has no impact for the comparison of the accuracy or speed of all procedures. Thus, we simply set $\bpi := m^{-1} \bone_m$. Because the computation speed of QATS depends on the expected number~$s$ of segments of the true state sequence, we study a selection of number of segments. Precisely, we are interested in settings with $s \in 1 + \{1, 2, 5, 10, 20, 50, \ldots\}$, up to~$s \le n/50$, since afterwards change points are too frequent in order for QATS to proceed efficiently.

We study the { homogeneous case of transition matrices~$\p^{(k)} = \p$ for all $k$.} We set~$\p$ to be the~$m\times m$ matrix with entries
\[
    p_{ij}
    \ := \
    \begin{cases}
        (m-1)^{-1}p & \text{if}\ i \neq j, \\
        1 - p & \text{if}\ i = j,\\
    \end{cases}
\]
where
$
    p = 
    p(n,s)
    :=
    (s-1)/(n-1)
$
is the \textit{exit probability}, i.e.\ the probability to transition to a different state than the current one. A Markov chain with such a transition matrix has an expected number~$s$ of segments, see Section~F of the Supplement. Considering sample sizes that are powers of~$10$ with an additional~$1$ helps with numerical stability.

As mentioned in the introduction, the observable process~$\Y$ takes values in an arbitrary measurable space~$(\YY,\BB)$. But for an observed sequence~$\y\in \YY^n$, only the value of each emission density~$f_i^{{(k)}}$ evaluated at~$y_k$ matters for the estimation, see \cref{Sec:Linearization}. Hence, we set the following normal model:
$
    { f_i^{(k)}(y)} = f_i(y) := \phi\left((y-i)/\sigma\right)
$,
where~$\phi$ is the density of the standard normal distribution with respect to Lebesgue measure, and 
$
    \sigma  \in \{ 0.1, 1.0 \}
$.
Thus,~$f_i$ is the density of the normal distribution with mean~$i$ and standard deviation~$\sigma$.

\subsection{Data generation}
\label{Sec:DataGeneration}
%
A state-observation sequence from an HMM with parameters $n$, $m$, $s$ and $\sigma$ is generated inductively as follows: Sample from a categorical distribution with parameter~$\bpi$ to generate the first state~$x_1^o$. Then, for~$k \in 2\col n$, sample from a categorical distribution with parameter~$(p_{x_{k-1}^o,j})_{j\in 1\col m}$ to generate~$x_k^o$. Finally, for each~$k\in 1\col n$, sample from~$\NN(x_k^o, \sigma^2)$ to generate~$y_k$. This yields a true state sequence~$\x^o = \x^o_{1\col n}$ with observed sequence~$\y = \y_{1\col n}$.

\subsection{Implementation}
\label{Sec:Implementation}
%

For completeness, the pseudocode of Viterbi is given in Algorithm~\ref{Alg:Viterbi}. It takes as an input the componentwise logarithm of the initial distribution~$\bpi$ and the transition matrix~$\p$, and a matrix~$\g\in \R^{m\times n}$ whose~$(i,k)$ entry is~$\log f_i(y_k)$. The PMAP algorithm requires forward and backward recursions which we implement as in \cite{Rabiner_1989} and its erratum.

\begin{algorithm}
\caption{Viterbi algorithm: $\mathrm{Viterbi}(\log \bpi, \log \p, \g)$}
\label{Alg:Viterbi}
\vspace{1ex}
\KwIn{$\log \bpi$, $\log \p$, $\g:=\bigl(\log f_i(y_k)\bigr)_{i\in\XX,k\in 1\col n}$}
\KwOut{$\hat{\bs x}$, a Viterbi path}
\lFor{$i\in 1\col m$}{
    $\rho_{i1} \gets \log \pi_i + g_{i1}$}
\For{$k\in 2 \col n$}{
    \For{$i \in 1\col m$}{
        \lFor{$j \in 1\col m$}{
            $\rho^{\mathrm{temp}}_j \gets \rho_{j,k-1} + \log p_{ji}$}
        $\zeta_{i, k-1} \gets \argmax_{j\in 1\col m} \rho^{\mathrm{temp}}_j$\;
        $\rho_{ik} \gets \max_{j \in 1 \col m} \rho^{\mathrm{temp}}_j + g_{ik}$\;}}
$\hat{x}_{n} \gets \argmax_{j \in 1 \col m} \rho_{jn}$\;
\lFor{$k=n-1,\ldots,1$}{
    $\hat{x}_k \gets \zeta_{\hat{x}_{k+1}, k}$}
\end{algorithm}

For a fair comparison between all methods, the parameters~$\log\bpi$ and $\log\p$ and the data~$\y$ are preprocessed outside of the timed computations. The matrix~$\g$ of log-densities and the matrices~{$\log \P$ and $\log\F$}  of cumulative log-densities, { together with~$\bar{\P}$ and~$\bar{\F}$,} are therefore precomputed from~$\y$. Precomputing the data as such could be performed while the collected data are being stored on the machine, or even instead of it. Furthermore, any temporary vector or matrix needed in QATS, Viterbi or PMAP, and whose size depends on~$n$ is declared outside of the timed computations and are passed as reference to their respective methods. Unlike described in Algorithm~\ref{Alg:QATS}, the computation of the final QATS-path from~$\SS$ and~$\hat{\z}$ in \eqref{Eq:QATS_path} counts for the computation time of QATS. We use the following optimization parameters for QATS: $\nu = 0.5$, $d_o=3$, $v_o=20$ and $\ns = 3$.

\subsection{Results}
\label{Sec:Results}
%
\paragraph{Computation time} 
To compare computation times for each setting~$(n,m,s,\sigma)$, we compute sample~$\beta$-quantiles of $T^{\mathrm{Q}}$, 
    $T^{\mathrm{V}}$,
    $T^{\mathrm{P}}$,
    $T^{\mathrm{V}}/T^{\mathrm{Q}}$
    and $T^{\mathrm{P}}/T^{\mathrm{Q}}$
from~$n_{\mathrm{sim}}=10^4$ independent repetitions, for~$\beta\in\{0.1,0.5,0.9\}$, where~$T^{\mathrm{Q}}$, $T^{\mathrm{V}}$ and $T^{\mathrm{P}}$ denote the computation times in seconds of QATS, Viterbi and PMAP, respectively. In particular, the ratios~$T^{\mathrm{V}}/T^{\mathrm{Q}}$ and $T^{\mathrm{P}}/T^{\mathrm{Q}}$ give the acceleration provided by QATS in comparison to Viterbi and PMAP.

\begin{figure}[!t]
\centering
\includegraphics[width=0.8\textwidth]{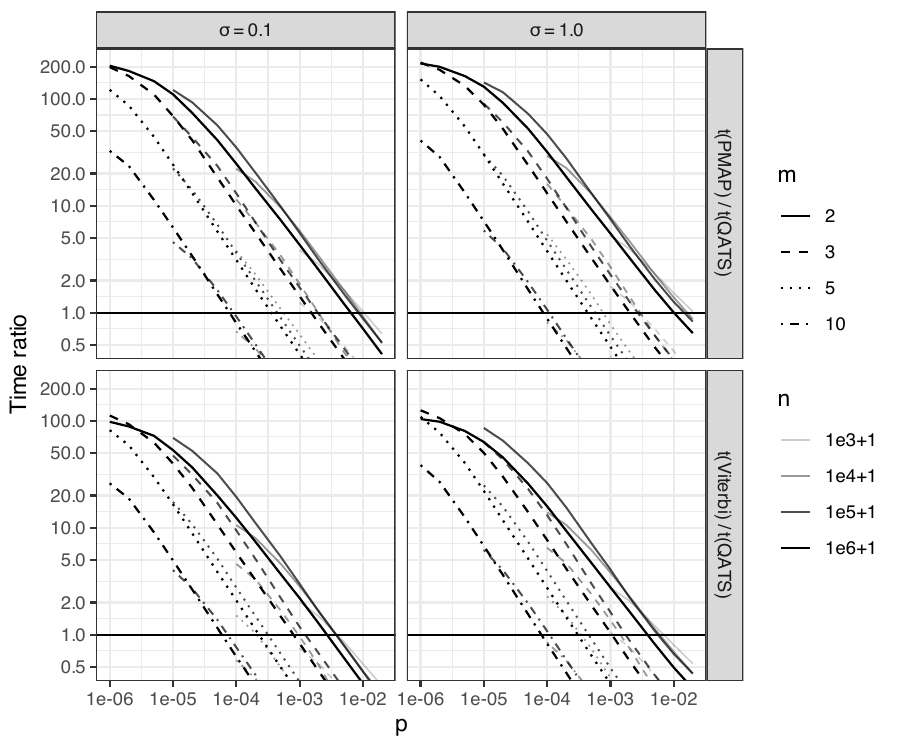}
\caption{Median time ratios against exit probability~$p$ using $\log$-scales.}
\label{Fig:Time_Ratios}
\end{figure}

Given the respective complexities of QATS, Viterbi and PMAP, we expect the ratios~$T^{\mathrm{V}}/T^{\mathrm{Q}}$ and $T^{\mathrm{P}}/T^{\mathrm{Q}}$ to behave as~$(pm\log n)^{-1}$. \cref{Fig:Time_Ratios} shows plots of time ratios against exit probabilities~$p$ with log-scales in both variables. We observe an almost negative linear relationship between the log-ratio and the log-probability. For fixed~$p$, time ratios increase if either~$m$ or~$n$ decrease. Those plots show that the standard deviation~$\sigma$ has little impact on estimation time, and that Viterbi is generally faster than PMAP.

When~$m=2$, accelerations of about~$30$ units are possible when~$p=10^{-4}$. When~$n=10^6+1$, it means that QATS is about~$30$ times faster than Viterbi when the expected number of segments~$s$ is~$101$, and about~$100$ times faster when $s=11$. If~$m=3$, $5$ or $10$, those ratios are smaller than for~$m=2$, but still often larger than~$2$, even for rather large values of~$p$. When~$m=10$ and~$n=10^6+1$, QATS is about~$5$ times faster than Viterbi when~$s=11$. This shows that QATS is substantially faster than Viterbi and PMAP for low number of segments/change points. Furthermore, time ratios could be even larger for even longer sequences of observations.

\begin{figure}[!t]
\centering
\includegraphics[width=0.9\textwidth]{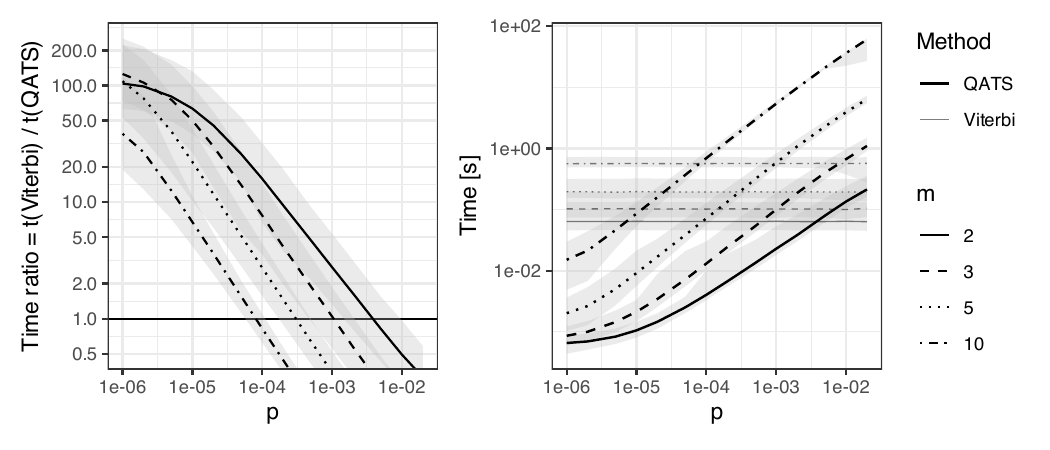}
\caption{Median quantile time ratio (left) and times (right) on $\log$-scales, with~$10\%$- and~$90$\%-quantile curves, for~$\sigma = 1.0$ and~$n = 10^6+1$.} 
\label{Fig:Time_Ratios_Bands}
\end{figure}

{
\cref{Fig:Time_Ratios_Bands} provides an assessment of the spread of time ratios and times for the setting~$\sigma = 1.0$ and~$n=10^6+1$. The~$10\%$- and~$90$\%-quantile curves display the expected variability around the median. The plot on the right demonstrate that the computation times of Viterbi is indeed independent of~$p$ (or the expected number of segments), whereas the log-computation time of QATS depends linearly on~$\log p$, except for small values of~$p$. Both methods are increasingly slower with increasing $m$.
}
\paragraph{Accuracy}
To evaluate the quality of an estimate~$\hat{\x}$, we compare it with the ground truth~$\x^o$. This true hidden path can indeed be accessed for this simulation study since data are generated. The comparison is done in terms of~$\ell^w$-type distances between~$\hat{\x}$ and~$\x^o$. Precisely, if~$\hat{\x}$ and~$\x^o$ are of length~$n$, we define
\[
    d_w(\hat{\x}, \x^o)
    \ := \
    \begin{cases}
        \frac{1}{n}
        \sum_{k\in 1\col n} 1_{\hat{x}_k \neq x_k^o} 
        & \text{if}\ w = 0,\\
        \left(
            \frac{1}{n}
            \sum_{k\in 1\col n} |\hat{x}_k - x_k^o|^{w} 
        \right)^{1/w} 
        & \text{if}\ w > 0.
    \end{cases}
\]
Hence, the quantity~$d_0(\hat{\x}, \x^o)$ corresponds to the proportion of misclassified (or misestimated) states, or simply \textit{misclassification rate}. On the other hand,~$d_w(\hat{\x}, \x^o)$ for~$w>0$ gives a measure of the amplitude of misclassifications scaled to the vector length. For~$w=2$, this is simply the \textit{root mean squared error}. Thus, for each setting~$(n,m,s,\sigma)$, we compute sample~$\beta$-quantiles of $d_w(\hat{\x}^{\mathrm{Q}}, \x^o)$, $d_w(\hat{\x}^{\mathrm{V}}, \x^o)$, $d_w(\hat{\x}^{\mathrm{P}}, \x^o)$, $d_w(\hat{\x}^{\mathrm{Q}}, \x^o) - d_w(\hat{\x}^{\mathrm{V}}, \x^o)$ and $d_w(\hat{\x}^{\mathrm{Q}}, \x^o) - d_w(\hat{\x}^{\mathrm{P}}, \x^o)$, for~$\beta\in\{0.1,0.5,0.9\}$ and~$w\in\{0,2\}$, where~$\hat{\x}^{\mathrm{Q}}$, $\hat{\x}^{\mathrm{V}}$ and~$\hat{\x}^{\mathrm{P}}$ correspond to QATS, Viterbi and PMAP paths, respectively.

\begin{figure}[!t]
\centering
\includegraphics[width=1\textwidth]{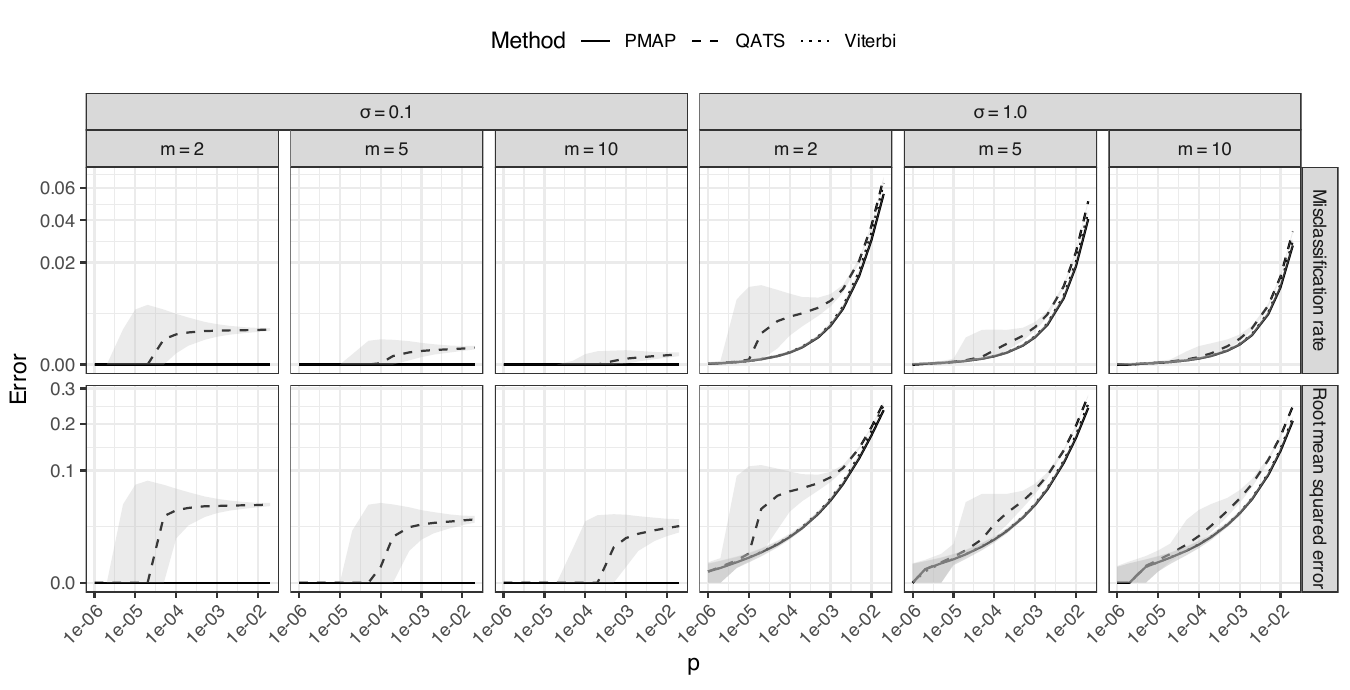}
\caption{Median error rates and~$10\%$- and~$90$\%-quantile curves in the square-root scale for the setting~$n=10^6+1$.}
\label{Fig:Errors}
\end{figure}

{
\cref{Fig:Errors} compares error rates of QATS, Viterbi and PMAP when~$n=10^6+1$. PMAP generally has the lowest error values, followed closely by Viterbi, and then QATS. When $\sigma = 0.1$, PMAP and Viterbi essentially perform errorlessly, unlike QATS whose error is small, but present. In the setting $\sigma = 1.0$, the three methods are hardly distinguishable when $m = 10$. Let us now comment on the interesting setting of~$m=2$ and~$\sigma=1.0$. For values of~$p$ smaller than~$10^{-5}$, all methods present small errors, with more variability for QATS. Then, for values of~$p$ in the interval~$[10^{-5},10^{-3}]$, error rates of Viterbi and PMAP increase, whereas the error of QATS takes a larger step up before flattening again. In this region, QATS is less able to differentiate between true variability (i.e.\ due to a large~$\sigma$) and variability due to a large number of segments (i.e.\ a large~$p$), than the two other methods. When~$p$ exceeds~$10^{-3}$, all methods interpret larger numbers of segments as extra noise, so the increase in error rates have similar behaviors again.}

\begin{figure}[!t]
\centering
\includegraphics[width=1\textwidth]{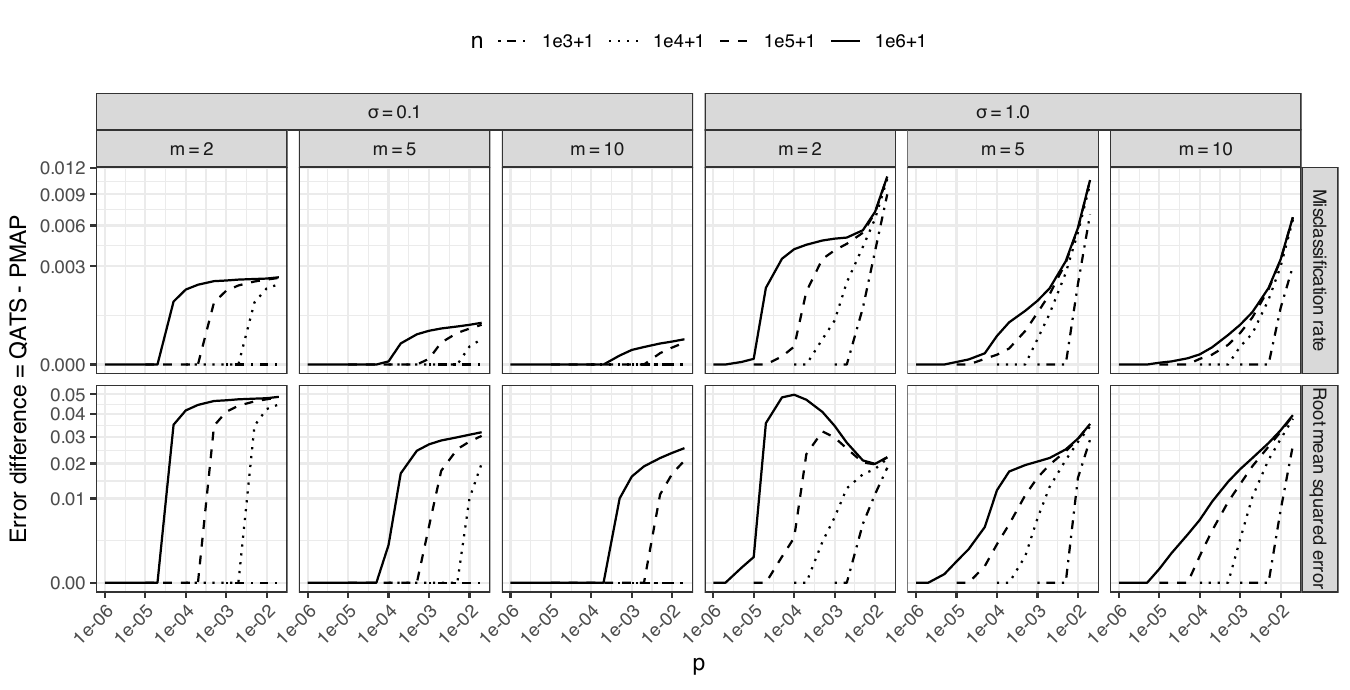}
\caption{Median difference in error rates (square-root scale) between PMAP and QATS.} 
\label{Fig:Errors_Differences}
\end{figure}

\cref{Fig:Errors_Differences} shows median error differences between QATS and PMAP only, as PMAP~is~the closest competitor to QATS. Decreasing~$n$ or $\sigma$, or increasing $m$, have the effect of reducing error differences. When~$\sigma = 0.1$, decreasing~$p$ (and therefore~$s$) lowers the error difference, but when~$\sigma = 1.0$ and~$m=2$, this is no longer true for the root mean squared error. Interestingly, QATS' and PMAP's number of misestimated states differ by at most $\approx 1\%$ (in median), no matter the setting considered. This allows us to conclude that QATS and PMAP have comparable misclassification rates.

\begin{figure}[!t]
\centering
\includegraphics[width=0.60\textwidth]{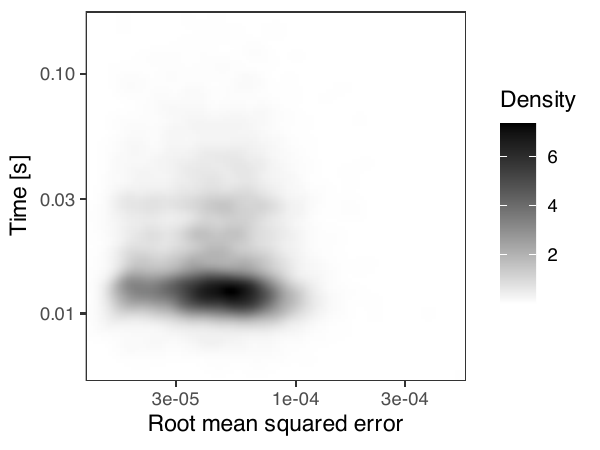}
\caption{Computation time of QATS against its root mean squared error, for~$n=10^6+1$, $m=3$, $p=10^{-4}$ and $\sigma = 1.0$. A kernel density estimator was used to build the heat map.} 
\label{Fig:Time_Errors}
\end{figure}

\paragraph{Time and error study}
One may suspect that fast computation times of QATS could be due to an error which would skip a substantial amount of steps in the procedure, but as indicated by \cref{Fig:Time_Errors}, this is not the case. When plotting each of the~$n_\mathrm{sim}$ measurements of log-time against log-root mean squared error in the setting~$n=10^6+1$, $m=3$, $p=10^{-4}$ and $\sigma = 1.0$, the two variables appear to be independent.

\paragraph{Misspecification and robustness}
{ We consider scenarios with model misspecification to reflect realistic settings where the model only holds approximately and parameters are estimated with error. Simulations in \Cref{sec:miss} demonstrate that QATS not more sensitive to model misspecification than Viterbi.} In addition, a robustness study with~$t$-distributed errors showed that Viterbi is no more robust than QATS. This is not surprising, as all methods share the same input: densities evaluated at each observation, and attempt to maximize certain likelihood scores.

{
\section{Real data analysis}
\label{Sec:Data}
%
\begin{figure}[!t]
\centering
\includegraphics[width=1\textwidth]{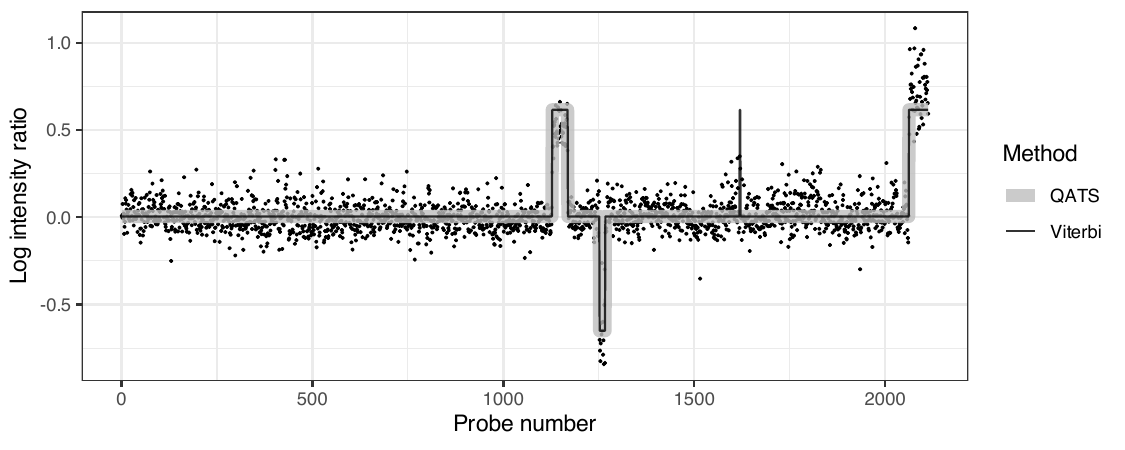}
\caption{{ Array CGH data for Coriell cell line 05296, segmented by QATS and Viterbi.}} 
\label{Fig:aCGH}
\end{figure}

The data shown in Figure~\ref{Fig:aCGH} correspond to array CGH (Comparative Genomic Hybridization) log-intensity ratios from the publicly available Coriell dataset 
provided by \citet{snijders2001assembly}. These data represent normalized hybridization signals across ordered genomic probes for a single Coriell cell line (05296), and are commonly used to benchmark methods for copy number variation detection. After parameter estimation using Baum–Welch algorithm, both Viterbi and QATS yield plausible state sequences, identifying chromosomal regions with similar underlying copy number states. Notably, while Viterbi path captures a fine-scale fluctuation (around~$k=1620$), QATS prefers a smoother segmentation, effectively flattening this subtle variation. Both interpretations appear reasonable and highlight the trade-off between sensitivity and parsimony in state sequence estimation.
}

\section{Discussion}
\label{Sec:Diss}
%
The proposed QATS is a greedy algorithm in nature. This complicates its theoretical justification, but eases its extensions to setups beyond HMMs (e.g.\ high-dimensional time series \citep{wang2019localizing,rinaldo2021localizing}) and to alternative likelihood functionals (e.g.\ including total variation as regularization \citep{wei2021inferring}). { To further accelerate QATS, one could alleviate the impact of the cubic complexity in the size $m$ of the state space by swapping the order of maximizations of local likelihoods and by using parallel computing. Precisely, one could first apply OS to maximize the local likelihood of a path with given states $\bs i$ but unknown change point(s) $k$ or $\bs k$. Provided those computations are split on $m^3$ threads, collecting the local maxima for each $\bs i$ and computing their maximum takes $\OO(m^3)$ operations. Because OS for local likelihood maximization and the maximum over all states $\bs i$ happen serially, this should result in computational complexity $\OO(s \max\{\log n, m^3\})$ for QATS.} The exploration of these extensions are promising avenues for future research. 

\section*{Acknowledgments}
%
The authors are grateful to Solt Kov\'acs and Yannick Baraud for stimulating discussions. The authors thank the editor, an associate editor, and two anonymous reviewers for their suggestions on an earlier version of the manuscript. HL and AM are funded by the Deutsche Forschungsgemeinschaft (DFG, German Research Foundation) under Germany’s Excellence Strategy--EXC 2067/1-390729940, and DFG Collaborative Research Center 1456. AM further acknowledges the support of DFG Research Unit 5381.


\clearpage

\appendix
\section{Proof of Lemma~\ref{Lem:aOS}}
%
The assumption on $V$ implies that all of its two-dimensional local maxima lie on the $s\times s$ grid $\KK \times \KK$. Further, by Lemma~\ref{Lem:OS}, the alternation of OS arrives on this grid after at most two searches, and remains on the grid afterwards. If the alternating procedure stops at a point, this point is a two dimensional local maximum of $V$, since it is a vertical and a horizontal maximum of $V$. Otherwise, by Lemma~\ref{Lem:OS}, the alternating procedure moves always to a point with a strictly larger value of $V$, implying that no loop (of at least two points) can occur in the alternating procedure and that a chain of alternation has length at most $s(s+1)/2$. The assertion of the Lemma follows, since each search requires at most $\OO\bigl(\log (r-\ell)\bigr)$ probes of $V$. \qed

\section{An alternative version of Algorithm~\ref{Alg:sOSH3}}
\label{App:AltAlgo}
%
The local maxima of $\HH^3$ at the boundary of $\KK^3_{\ell \col r}$ correspond to the local maxima of $\HH^2$ on $\KK^2_{\ell \col r}$. Thus, it is more interesting to find local maxima of $\HH^3$ in the interior of $\KK^3_{\ell \col r}$. To this end, instead of maximizing $\HH^2$, $\KK^2_{\ell:(k_2-1)} \ni k\mapsto \HH^3(k,k_2)$ and $\KK^2_{k_1:r} \ni k\mapsto \HH^3(k_1,k)$, it may be desirable to rotate each score so that their value on both ends of their respective domains coincides. Precisely, we would set
\begin{equation*}
	\bar{\HH}^2(k) \ := \
	\HH^2(k) - \bigl(\HH^2(r)-\HH^2(\ell+1)\bigr) \frac{k-\ell-1}{r-\ell-1},
\end{equation*}
for $k\in\KK^2_{\ell:r}$,
\begin{equation*}
    \bar{\HH}^3(k, k_2) \ := \
	\HH^3(k, k_2) 
	 - \bigl(\HH^3(k_2-1,k_2)-
	\HH^3(\ell+1,k_2)\bigr) \frac{k-\ell-1}{k_2-\ell-2},
\end{equation*}
for $k\in\KK^3_{\ell:(k_2-1)}$, and
\begin{equation*}
    \bar{\HH}^3(k_1, k) \ := \
	\HH^3(k_1, k) 
	- \bigl(\HH^3(k_1,r)-
	\HH^3(k_1,k_1+1)\bigr) \frac{k-k_1-1}{r-k_1-1},
\end{equation*}
for $k\in\KK^3_{k_1:r}$, whenever the above fractions are respectively well-defined, that is whenever the sets $\KK^2_{\ell:r}$, $\KK^3_{\ell:(k_2-1)}$ and $\KK^3_{k_1:r}$ are respectively non-empty. Then, one would replace $H^2$ by $\bar{\HH}^2$ in Algorithm~\ref{Alg:OSH2}, $H^3(k_1^*,\cdot)$ and $H^3(\cdot,k_2^*)$ by $\bar{\HH}^3(k_1^*,\cdot)$ and $\bar{\HH}^3(\cdot,k_2^*)$ in Algorithm~\ref{Alg:sOSH3}.

\section{Log-local likelihoods}
\label{App:LogLocLike}
%
{
Recall the definition of matrices $\F$, $\bar{\F}$ $\P$ and $\bar{\P}$ in Section~\ref{Sec:Linearization}, and introduce the entry-wise logarithms  $\QQ := \log \P$ and $\G := \log \F$ as follows:
\begin{equation*}
    \QQ
	=
	(Q_{ik})
	:=
	\left(
		\sum_{t = 1}^k q_{ii}^{(t)}
	\right)
	\;\text{ and }\; 
	\G
	=
	\left(G_{ik}\right)
	:=
	\left(
		\sum_{t=1}^k \log \bigl(f_i^{(t)}(y_t) + 1_{\left[f_i^{(t)}(y_t) = 0\right]}\bigr)
	\right),
\end{equation*}
where $q_{ij}^{(t)} = \log \bigl(p_{ij}^{(t)} + 1_{\left[p_{ij}^{(t)} = 0\right]}\bigr)$. We adopt the convention that $\bar{F}_{i0} = \bar{P}_{i0} = G_{i0} = 0$ and $q_{ij}^{(0)} = 0$ for $i,j\in\XX$. Now we may define
\[
\GG_{\ell:r}(i|x_0)
	\ := \
	G_{i,r} - G_{i,\ell-1}   + 1_{[\ell=1]}\log\pi_{i} 
	+  q_{x_0i}^{(\ell-1)} + Q_{i,r-1} -  Q_{i,\ell-1},
\]
if $\bar{F}_{i,\ell-1} = \bar{F}_{i,r}$ and $\bar{P}_{i,\ell-1} = \bar{P}_{i,r-1}$; and otherwise $\GG_{\ell:r}(i|x_0) := -\infty$. Similarly, we define 
\begin{multline*}
\GG_{\ell:k:r}(\bi|x_0)
	\ := \
	 G_{i_1,k-1} - G_{i_1,\ell-1} 
	+ G_{i_2,r}  - G_{i_2,k-1} + 1_{[\ell=1]}\log\pi_{i_1} \\
		+ q_{x_0i_1}^{(\ell-1)} 
	 	+ Q_{i_1,k-2}-Q_{i_1,\ell-1} + q_{i_1i_2}^{(k-1)} + Q_{i_2,r-1} - Q_{i_2,k-1},
\end{multline*}
if  $\bar{F}_{i_1,\ell-1} = \bar{F}_{i_1,k-1}$, $\bar{F}_{i_2,k-1} = \bar{F}_{i_2,r}$, $ \bar{P}_{i_1,\ell-1} = \bar{P}_{i_1,k-2}$ and $\bar{P}_{i_2,k-1} = \bar{P}_{i_2,r-1}$; and otherwise $\GG_{\ell:k:r}(\bi|x_0) := -\infty$. Also, we define
\begin{multline*}
\GG_{\ell:k_1:k_2:r}(\bi|x_0)
	:=
	  G_{i_1,k_1-1} - G_{i_1,\ell-1} 
	+ G_{i_2,k_2-1} 
	  - G_{i_2,k_1-1}	+ G_{i_3,r} - G_{i_3,k_2-1}  + 1_{[\ell=1]}\log\pi_{i_1}
	   + q_{x_0i_1}^{(\ell-1)} \\
	   	+ Q_{i_1,k_1-2}-Q_{i_1,\ell-1}+ q_{i_1i_2}^{(k_1-1)} 
	  + Q_{i_2,k_1-1} - Q_{i_2,k_2-2}+ q_{i_2i_3}^{(k_2-1)} + Q_{i_3,r-1} - Q_{i_3,k_2-1},
\end{multline*} 
if $ \bar{F}_{i_1,\ell-1} = \bar{F}_{i_1,k-1}$, $\bar{F}_{i_2,k_1-1} = \bar{F}_{i_2,k_2-1}$, $ \bar{F}_{i_3,k_2-1} = \bar{F}_{i_3,r}$, $\bar{P}_{i_1,\ell-1} = \bar{P}_{i_1,k_1-2}$, $\bar{P}_{i_2,k_2-2} = \bar{P}_{i_2,k_1-1}$ and $\bar{P}_{i_3,k_2-1} = \bar{P}_{i_3,r-1}$; and otherwise $\GG_{\ell:k_1:k_2:r}(\bi|x_0) := -\infty$. 
Then, we have
\begin{align*}
	\log \Lambda_{\ell:r} (i \bone_{r-\ell+1}|x_0 )
	\ &= \
	\GG_{\ell:r}(i|x_0), \\
	\log \Lambda_{\ell:r} \bigl((i_1 \bone_{k-\ell},i_2 \bone_{r-k+1})|x_0 \bigr)
	\ &= \
	\GG_{\ell:k:r}(\bi|x_0), \\
	\log \Lambda_{\ell:r} \bigl((i_1 \bone_{k_1-\ell},i_2 \bone_{k_2-k_1},i_3 \bone_{r-k_2+1})|x_0 \bigr)
	\ &= \
	\GG_{\ell:k_1:k_2:r}(\bi|x_0).
\end{align*}
}
\section{Proofs for Section~\ref{Sec:Justification}}
%
We progressively build the theory to prove Lemma~\ref{Lem:Charac_CP}. We suppose that Assumption~\ref{Aspt:SimpleSetting} holds for the remainder of the appendix, and define the following constants for $c\in 1\col 3$:
\[
	\beta^c_1
	\ := \
	- \log 2 + (n-c)\log(1-\varepsilon) + (c-1) \log \varepsilon + n\beta_1 + \frac{n}{2}\beta_2,
\]
as well as the following elements for $k\in \KK^2$ and $\k\in \KK^3$:
\begin{align*}
	\psi^1
	\ &:= \
	\sum_{t=1}^n 1_{[y_t = 1]},&\\
	\psi^2(k)
	\ &:= \
	\sum_{t=1}^{k-1} 1_{[y_t = 1]}
	+ \sum_{t=k}^{n} 1_{[y_t = 2]},
	\quad \text{for}\ k\in \KK^2,\\
	\psi^3(\k)
	\ &:= \
	\sum_{t=1}^{k_1-1} 1_{[y_t = 1]}
	+ \sum_{t=k_1}^{k_2-1} 1_{[y_t = 2]}
	+ \sum_{t=k_2}^{n} 1_{[y_t = 1]}.
\end{align*}

\begin{Lemma}
\label{Lem:HH2_HH3_Char}
We have that
\[
	\HH^c
	\ = \
	\beta_1^c +
	\beta_2 \Bigl(
		\max\bigl( \psi^c, n-\psi^c \bigr) - \frac{n}{2}
	\Bigr)
	\ = \
	\beta_1^c + 
	\beta_2 \left|\psi^c - \frac{n}{2} \right|
\]
for $c\in 1\col 3$. Furthermore, for each segment $\KK^2_a:=(\ka_{a-1}\col \ka_{a})\cap \KK^2$, $a\in 1\col s^o$, there exists $\gamma_a\in\Z/2$ such that
$
	\HH^2(k) 
	=
	\beta_1^2 + \beta_2 \left|k - \gamma_a \right|,
$
for $k\in \KK^2_a$. Likewise, for each nonempty set of contiguous index pairs $\KK^3_{ab}:=\bigl((\ka_{a-1}\col \ka_a)\times (\ka_{b-1}\col \ka_b)\bigr)\cap \KK^3$, where $a,b\in 1\col s^o$, $a\le b$, there exists $\gamma_{ab}\in \Z/2$ such that
$
	\HH^3(\k)
	=
	\beta_1^3 + \beta_2 \left| k_2 - (-1)^{a+b} k_1 - \gamma_{ab} \right|,
$
for $\k \in \KK^3_{ab}$.
\end{Lemma}

\begin{proof}[\textbf{Proof of Lemma~\ref{Lem:HH2_HH3_Char}}]
The special forms of $g_i$ and the fact that $m=2$ imply that
\begin{align}
	&\GG_{1:n}(i)
	\ = \
	\beta_1^1 - \frac{n}{2}\beta_2+ \beta_2 \cdot 
	\begin{cases}
		\psi^1 	  & \text{if}\ i = 1, \\
		n - \psi^1 & \text{if}\ i = 2,
	\end{cases} \nonumber\\
	&\GG_{1:k:n}(\bi)
	\ = \
	\beta_1^2 - \frac{n}{2}\beta_2+ \beta_2 \cdot
	\begin{cases}
		\psi^2(k) 	   & \text{if}\ \bi = (1,2), \\
		n - \psi^2(k) & \text{if}\ \bi = (2,1),
	\end{cases}\nonumber \\
	&\GG_{1:k_1:k_2:n}(\bi)  \ = \
	\beta_1^3 - \frac{n}{2}\beta_2+ \beta_2 \cdot
	\begin{cases}
		\psi^3(\k)     & \text{if}\ \bi = (1,2,1), \\
		n - \psi^3(\k) & \text{if}\ \bi = (2,1,2),
	\end{cases}\label{Eq:G_HH3}
\end{align}
for $k\in\KK^2$ and $\k\in\KK^3$. In consequence, we find that
\[
	\HH^c
	=
	\beta_1^c + \beta_2
	\left(
		\max(\psi^c, n-\psi^c) - \frac{n}{2}
	\right)
	=
	\beta_1^c + \beta_2 \left| \psi^c -\frac{n}{2} \right|,
\]
for $c\in 1 \col 3$, where we used the fact that $2\max(\xi_1,\xi_2)=\xi_1 + \xi_2 + |\xi_1 - \xi_2|$ in the last equality.

For the rest of the proof, observe first that for $a \in 1\col s^o$ and $t\in \ka_{a-1}\col (\ka_a - 1)$, we have
\[
	y_t
	\ = \	
	\begin{cases}
		1 & \text{if}\ y_1+a\ \text{is even}, \\
		2 & \text{if}\ y_1+a\ \text{is odd}. \\
	\end{cases}
\]
For $a\in 1\col s^o$ and $k\in \KK^2_a$, there exists some $\alpha_{a}\in\Z$ such that
\begin{align*}
	\psi^2(k) & \ = \
	\sum_{t=1}^{n} 1_{[y_t = 1]}
	+ 1_{[k<\ka_a]}\sum_{t=k}^{\ka_a-1} (-1)^{y_t} + 1_{[a<s^o]} \sum_{t=\ka_a}^{n} (-1)^{y_t} \\
	& \ = \
	\alpha_{a} - (-1)^{y_1+a} (\ka_a-k).
\end{align*}
Since $\gamma_a := -\ka_a + (-1)^{y_1+a} (\alpha_{a} - n/2)$ satisfies $| \psi^2(k) - n/2 | = | k - \gamma_a |$, the second part of the Lemma is proved.

Let $a,b\in 1\col s^o$, $a\le b$, such that $\KK^3_{ab}\neq\emptyset$, and fix $\k\in\KK^3_{ab}$. In case $a=b$, then $a+b$ is even and $\ka_{b-1} \le k_1 \le k_2 -1 < \ka_b$, so for some for some $\alpha_{ab}\in \N$ we have that
\begin{align*}
	\psi^3(\k)
	\ =& \
	\sum_{t=1}^{n} 1_{[y_t = 1]} -(-1)^{y_1+b}(k_2-k_1)\\
	\ =& \
	\alpha_{ab} -(-1)^{y_1+b}\bigl(k_2-(-1)^{a+b}k_1\bigr).
\end{align*}
When $a<b$, we have $\ka_{a-1}\le k_1 \le \ka_a\le \ka_{b-1}\le k_2 \le \ka_b$ and for some $\alpha_{ab}\in \Z$ it holds that
\begin{align*}
	\psi^3(\k)
	=\ &
	\sum_{t=1}^{n} 1_{[y_t = 1]}
	+ 1_{[k_1< \ka_a]}\sum_{t=k_1}^{\ka_a-1} (-1)^{y_t}\\
	& \ + 1_{[a+1<b]}\sum_{t=\ka_a}^{\ka_{b-1}-1} (-1)^{y_t}
	+ 1_{[\ka_{b-1}<k_2]}\sum_{t=\ka_{b-1}}^{k_2-1} (-1)^{y_t} \\
     =\ &
	\sum_{t=1}^{n} 1_{[y_t = 1]}
	- (-1)^{y_1+a} (\ka_a-k_1) \\
	& \ + 1_{[a+1<b]}\sum_{t=\ka_a}^{\ka_{b-1}-1} (-1)^{y_t}
	- (-1)^{y_1+b} (k_2-\ka_{b-1})\\
	=\ &
	\alpha_{ab} - (-1)^{y_1+b} 
	\bigl(
		k_2 - (-1)^{a+b} k_1	
	\bigr).
\end{align*}
In both cases, $\gamma_{ab} := (-1)^{y_1+b} (\alpha_{ab} - n/2)$ satisfies $|\psi^3(\k) - n/2| = |k_2 - (-1)^{a+b}k_1 - \gamma_{ab} |$ and thus the last part of the Lemma is proved.
\end{proof}

\begin{proof}[\textbf{Proof of Lemma~\ref{Lem:Charac_CP}}] \textbf{Part~(i).}
A local maximum of $\HH^2$ on $\KK^2$ belonging to a set $\KK^2_a$, $a\in 1\col s^o$, necessarily has to be a local maximum of $\HH^2$ restricted to that $\KK^2_a$. But the structure of $\HH^2$ and the fact that $\beta_2>0$ imply that local maxima of $\HH^2$ restricted to $\KK^2_a$ can only be located at its endpoints $\max(2,\ka_{a-1})$ and $\min(\ka_a,n)$. Likewise, a local maximum of $\HH^3$ on $\KK^3$ belonging to a set $\KK^3_{ab}$, $a,b\in 1\col s^o$ and $a\le b$, necessarily has to be a local maximum of $\HH^3$ restricted to that $\KK^3_{ab}$. The structure of $\HH^3$ and the fact that $\beta_2>0$ imply that local maxima of $\HH^3$ restricted to $\KK^3_{ab}$ are necessarily distributed as follows:

If $a+b$ is even, then $\HH^3(\k)=\beta^3_1+\beta_2|k_2-k_1-\gamma_{ab}|$ on $\KK^3_{ab}$. When $a=b$, local maxima of $\HH^3$ restricted to $\KK^3_{ab}$ can only be located on the diagonal $\KK^3_{ab}\cap \{\k\in\KK^3:k_2= k_1+1\}$ or at $\bigl(\max(2,\ka_{a-1}),\min(\ka_b,n)\bigr)$. When $a<b$, then $a < b-1$ (because $a+b$ is even) and only $\bigl(\max(2,\ka_{a-1}),\min(\ka_b,n)\bigr)$ or $\bigl(\ka_a,\ka_{b-1}\bigr)$ can be local maxima of $\HH^3$ restricted to $\KK^3_{ab}$.

If $a+b$ is odd, then $a<b$ and $\HH^3(\k)=\beta^3_1+\beta_2|k_2+k_1-\gamma_{ab}|$ on $\KK^3_{ab}$. That means, local maxima of $\HH^3$ restricted to $\KK^3_{ab}$ can only be located at $\bigl(\max(2,\ka_{a-1}),\ka_{b-1}\bigr)$ or $\bigl(\ka_a,\min(\ka_b,n)\bigr)$.

All in all, the sum of indices indexing local maxima of $\HH^3$ in the interior of $\KK^3$ is always odd. 

\bigskip

\textbf{Part~(ii).}
Let $a \in 1 \col (s^o-1)$ be arbitrary. We suppose without loss of generality that $y_{\ka_a} = 1$. It is then clear to see that
\begin{equation}
    \label{Eq:HH3_NP}
    \psi^3(\ka_{a-1}, \ka_{a}) 
    \ = \ \psi^3(\ka_a, \ka_{a+1}) + (\ka_{a+1} - \ka_{a-1}).
\end{equation}
Furthermore, $2\psi^3(\ka_{a-1},\ka_a) \ge n + 2$ or $2\psi^3 (\ka_a, \ka_{a+1}) \le n-2$, because otherwise the relation in \eqref{Eq:HH3_NP} implies that $\ka_{a+1} - \ka_{a-1} < 2$, which is a contradiction.  

\begin{itemize}
\item In case $2\psi^3(\ka_{a-1},\ka_a) \ge n + 2$, we apply \eqref{Eq:G_HH3} and obtain 
\[
    \HH^3 (k_1,k_2) 
    \ = \ \GG_{1:k_1:k_2:n}(1,2,1) 
    \ \ge \ \GG_{1:k_1:k_2:n}(2,1,2)
\]
for $(k_1, k_2) \in \KK^3$ and $\|(k_1,k_2) - (\ka_{a-1},\ka_{a})\|_1 \le 1$. Thus, the pair $(\ka_{a-1}, \ka_{a})$ is a local maximum of $\HH^3$.

\item In case $2\psi^3 (\ka_a, \ka_{a+1}) \le n-2$, by \eqref{Eq:G_HH3} we obtain 
\[
    \HH^3 (k_1,k_2) 
    \ = \ \GG_{1:k_1:k_2:n}(2,1,2) 
    \ \ge \ \GG_{1:k_1:k_2:n}(1,2,1)
\]
for $(k_1,k_2) \in \KK^3$ and $\|(k_1,k_2) - (\ka_{a},\ka_{a+1})\|_1 \le 1$. Thus, the pair $(\ka_{a}, \ka_{a+1})$ is a local maximum of $\HH^3$.
\end{itemize}

Recall that at least one of the two above cases is valid, which concludes the proof.
\end{proof}

\begin{Remark}
An intuition for the fact that $a+b$ must be odd is as follows: The search for the best local path with three segments looks for bumps, that is a state sequence $(1,2,1)$ or $(2,1,2)$ with jumps occurring at some $2\le \ka_a<\ka_b\le n$. But for such a feature to exist in $\x^o$, it is necessary for $a$ and $b$ to not share the same parity, i.e.\ $a+b$ must be odd.
\end{Remark}

\section{Proofs for Section~\ref{Sec:Sensitivity}}
%
We start with some technical preparations. Recall that $\ka_{s^o} = n+1$ and define, when $s^o\ge 2$, the numbers $w(a)_1,\ldots,w(a)_{s^o-1}$ to be the sorted elements of $(1\col s^o)\setminus \{a\}$ and when $s^o\ge 3$, the numbers $w(a,b)_1,\ldots,w(a,b)_{s^o-2}$ to be those of $(1\col s^o)\setminus \{a,b\}$. Furthermore, we let
\[
	\varphi^c
	\ := \
	\left\lceil \frac{s^o-c}{2} \right\rceil
	\qquad
	\text{for}\ c\in 1\col 3,
\]
and set to $0$ any sum whose index of summation ranges from $1$ to $0$.

\begin{Lemma}
\label{Lem:HHc}
If $s^o\ge 1$, we have that
\[
	\HH^1
	\ = \
	\beta_1^1 + \beta_2
	\left| \frac{n}{2} - \eta^1 \right|
\]
with $\eta^1
	:=
	\sum_{t=1}^{\varphi^1} 
	\left( \ka_{2t} - \ka_{2t-1} \right)$.
If $s^o\ge 2$ and $a\in 1\col (s^o-1)$, then
\[
	\HH^2(\ka_a)
	\ = \
	\beta_1^2 + \beta_2
	\left| \frac{n}{2} - \eta^2(a) \right|
\]
with $\eta^2(a)
	:=
	\sum_{t=1}^{\varphi^2} 
	\left( \ka_{w(a)_{2t}} - \ka_{w(a)_{2t-1}}\right)$.
If $s^o\ge 3$ and $a,b\in 1\col (s^o-1)$, $a<b$ and $a+b$ odd, then
\[
	\HH^3(\ka_a,\ka_b)
	\ = \
	\beta_1^3 + \beta_2
	\left| \frac{n}{2} - \eta^3(a,b) \right|
\]	
with $\eta^3(a,b) := 
	\sum_{t=1}^{\varphi^3}
	\left( \ka_{w(a,b)_{2t}} - \ka_{w(a,b)_{2t-1}} \right)$.
\end{Lemma}

\begin{proof}[\textbf{Proof of Lemma~\ref{Lem:HHc}}]
Since $|\psi^c-n/2|=|(n-\psi^c)-n/2|$ for all arguments of $\psi^c$ and all $c \in 1\col 3$, we may assume, without loss of generality, that $y_1=1$. Hence, $y_t$ is equal to $2$ on $\ka_{2t-1}\col (\ka_{2t}-1)$, for all $t\in 1\col \varphi^1$ (when $s^o>1$, otherwise $\varphi^1=0$ and $y_t$ is always $1$). Thus $\psi^1=n-\eta^1$.

For the results concerning $\HH^2$ and $\HH^3$, we distinguish between the four cases generated by the various combinations of $s^o$ even or odd and $a$ even or odd (and therefore $b$ odd or even). We only prove the result concerning $s^o$ odd and $a$ even. The proof of the other three cases is analogous.

When $s^o\ge 2$ is odd and $a\in 1\col (s^o-1)$ is even, we have
\begin{align*}
	\psi^2 & (\ka_a) 
	\ = \
	(\ka_1-1) + \cdots + (\ka_{a-1} - \ka_{a-2}) \\
	& \ + (\ka_{a+2} - \ka_{a+1}) + \cdots + (\ka_{s^o-1}-\ka_{s^o-2}) \\
	\ =& \
	- (\ka_2-\ka_1) - \cdots - (\ka_{a-2} - \ka_{a-3}) - (\ka_{a+1} - \ka_{a-1}) \\
	& \ - (\ka_{a+3} - \ka_{a+2}) - \cdots - (\ka_{s^o} - \ka_{s^o-1}) + n \\
	\ =& \
	n - \eta^2(a).
\end{align*}
When $s^o\ge 3$ is odd and $a,b\in 1\col (s^o-1)$ with $a<b$, $a$ even and $b$ odd, we have
\begin{align*}
	\psi^3&(\ka_a,\ka_b)
	\ = \
	(\ka_1-1) + \cdots + (\ka_{a-1} - \ka_{a-2}) \\
	& \ + (\ka_{a+2} - \ka_{a+1}) + \cdots + (\ka_{b-1} - \ka_{b-2}) \\
	& \ + (\ka_{b+2} - \ka_{b+1}) + \cdots + (\ka_{s^o-2} - \ka_{s^o-3}) \\
	& \ + (n - \ka_{s^o-1} + 1) \\
	\ =& \
	- (\ka_2-\ka_1) - \cdots - (\ka_{a-2} - \ka_{a-3}) - (\ka_{a+1} - \ka_{a-1}) \\
	& \ - (\ka_{a+3} - \ka_{a+2}) - \cdots - (\ka_{b-2} - \ka_{b-3}) - (\ka_{b+1} - \ka_{b-1}) \\
	& \ - (\ka_{b+3} - \ka_{b+2}) - \cdots - (\ka_{s^o-1} - \ka_{s^o-2}) + n \\
	\ =& \
	n - \eta^3(a,b).\qedhere
\end{align*}
\end{proof}

\begin{Lemma}
\label{Lem:HHc_so_123}
When $s^o=1$,
\begin{align*}
	\HH^1 
	\ &= \ \beta_1^1 + \beta_2 \frac{n}{2}, \\
	\HH^2(2)
	\ =  \ \HH^2(n)
	\ &= \ \beta_1^2 + \beta_2\left(\frac{n}{2}-1\right)
	\ >  \ \HH^2(k), 
\end{align*}
for $k\in\KK^2\setminus\{2,n\}$,
\begin{align*}
    \HH^3(k,k+1)
	\ &= \ \beta_1^3 + \beta_2\left(\frac{n}{2}-1\right)
	\ >  \ \HH^3(\k),
\end{align*}
for $k\in\KK^2\setminus\{n\},\, \k\in \KK^3, k_1 + 1 < k_2$. When $s^o=2$,
\begin{align*}
	\HH^1 
	\ &= \ \beta_1^1 + \beta_2 \left| \frac{n}{2} - \ka_1 + 1\right|, \\
	\HH^2(\ka_1)
	\ &= \ \beta_1^2 + \beta_2 \frac{n}{2}
	>  \ \HH^2(k), 
	\qquad k\in\KK^2\setminus\{ \ka_1\},\\
	\HH^3(2,\ka_1)
	\ &=  \ \HH^3(\ka_1,n)
	\ = \ \beta_1^3 + \beta_2\left(\frac{n}{2}-1\right)
	>  \ \HH^3(\k),
\end{align*}
for $\k\in\KK^3\setminus\{(2,\ka_1),(\ka_1,n)\}$, provided $(2,\ka_1)$ and $(\ka_1,n)$ are elements of $\KK^3$, otherwise the corresponding expression is omitted.

When $s^o=3$ and for $a\in 1\col 2$,
\begin{align*}
	\HH^1
	\ &= \ 
	\beta_1^1+\beta_2 \left|\frac{n}{2} - \ka_2 + \ka_1\right|, \\
	\HH^2(\ka_a)
	\ &= \
	\beta_1^2+\beta_2 \left|\frac{n}{2} - \ka_{w(a)_1} + 1 \right|
	> \HH^2(k),
\end{align*}
for $k\in \KK^2\setminus\{2,\ka_1,\ka_2,n\}$,
\begin{align*}
	\HH^2(2) \ = \ \HH^2(n)
	\ &= \
	\beta_1^2+\beta_2 \left|\frac{n}{2} - \ka_2 + \ka_1 -1 \right|
	> \HH^2(k),
\end{align*}
for $k\in \KK^2\setminus\{2,\ka_1,\ka_2,n\}$,
\begin{align*}
	\HH^3(\ka_1,\ka_2)
	\ &= \
	\beta_1^3  + \beta_2 \frac{n}{2}
	> \HH^3(\k),
\end{align*}
for $\k\in \KK^3\setminus\{(\ka_1,\ka_2)\}$.
\end{Lemma}

\begin{proof}[\textbf{Proof of Lemma~\ref{Lem:HHc_so_123}}]
Similarly as in the previous proof, since $|\psi^c-n/2|=|(n-\psi^c)-n/2|$ for all arguments of $\psi^c$ and all $c \in 1\col 3$, we may assume, without loss of generality, that $y_1=1$ and study maxima and minima of $\psi^c$ to infer on maxima of $|\psi^c-n/2|$

When $s^o=1$, $\psi^2(k)=k-1$, so $\HH^2(k)=\beta_1^2 + \beta_2 |n/2 - k + 1|$ with maximum $\beta_1^2 + \beta_2 (n/2 - 1)$ attained at $k=2$ and $k=n$. Likewise, $\psi^3(\k)=n-k_2+k_1$, so $\HH^3(\k)=\beta_1^3 + \beta_2 |n/2 - k_2 + k_1|$ with maximum $\beta_1^3 + \beta_2 (n/2 - 1)$ attained at $\k\in\KK^3$ such that $k_1+1=k_2$.

When $s^o=2$, $\psi^2$ is by definition maximal and equal to $n$ at $\ka_1$, whereas its minimal value is at least as large as $2$. Likewise, $\psi^3$ is maximal with value $n-1$ at $(\ka_1,n)$ if $\ka_1<n$ and minimal with values $1$ at $(2,\ka_1)$ if $2<\ka_1$. Since $n>2$, at least one of the two cases must hold, yielding a maximum of $\beta_1^3 + \beta_2 (n/2 - 1)$ for $\HH^3$.

When $s^o=3$, by Part~(i) of Lemma~\ref{Lem:Charac_CP}, the maximum of $\HH^2$ is attained either at the boundaries $2$ or $n$, or at the change points $\ka_1$ or $\ka_2$. In the first case, $\psi^2(2)=\ka_2-\ka_1 + 1$ and $\psi^2(n)=n-\ka_2+\ka_1-1$, yielding $\HH^2(2)=\HH^2(n)=\beta_1^2 + \beta_2 |n/2 - \ka_2 - \ka_1 - 1|$, and in the second case, $\psi^2(\ka_1)=\ka_2-1$ and $\psi^2(\ka_2)=\ka_1-1$, yielding $\HH^2(\ka_a)=\beta_1^2 + \beta_2 |n/2 - \ka_{w(a)_1} + 1|$, for $a\in1\col 2$. As to $\HH^3$, we have that $\psi^3$ is maximal equal to $n$ at $(\ka_1,\ka_2)$, and at least as large as $3$ otherwise.
\end{proof}

\begin{proof}[\textbf{Proof of Lemma~\ref{Lem:OptimH0H1H2}}]
The proof of each inequality derives either directly or in part from Lemma~\ref{Lem:HHc_so_123}. For instance, since $\beta_1^1-\beta_1^2=\beta_1^2-\beta_1^3= \beta_2\delta > 0$, inequalities concerning $s^o=1$ as well as the first inequality concerning $s^o=2$ are clear. When $s^o=2$, then $\HH^2(\ka_1) > \HH^1$ holds if, and only if, $n/2 - |n/2 - \ka_1 +1|> \delta$, which is equivalent to $\ka_1 \in (1+\delta,n+1-\delta)$. We consider the case $s^o=3$. Then $\HH^2(\ka_1)>\HH^1$ is equivalent to
\[
	\left|\frac{n}{2} - \ka_2 + 1\right|
	- \left|\frac{n}{2} - \ka_2 + \ka_1\right|
	\ > \
	\delta.
\]
The only possibility for the displayed inequality to hold is if $\ka_2> n/2 + 1$, since otherwise $\ka_2 \le n/2 + 1 < n/2 + \ka_1$ would imply that the displayed inequality simplifies to $\ka_1 < 1-\delta$, which is impossible. In consequence, either $\ka_2\le n/2+\ka_1$, in which case the displayed inequality is equivalent to $\ka_1<2\ka_2 - n - 1 - \delta$, or $\ka_2 > n/2+\ka_1$, in which case the displayed inequality is equivalent to $\ka_1>1+\delta$. Combining our findings, the displayed inequality is equivalent to $\ka_2 > n/2 + 1$ and $\ka_1\in(1+\delta,2\ka_2 - n - 1 - \delta)$. To conclude, the latter interval is non-empty if, and only if, $\ka_2>n/2+1+\delta$. The equivalence statements regarding $\HH^2(\ka_2)>\HH^1$ and $\HH^3(\ka_1,\ka_2)>\HH^1$ are proved with similar arguments.

We prove the equivalent statement to $\HH^3(\ka_1,\ka_2) > \max_{k\in\KK^2}\, \HH^2(k)$. To this end, note that the inequality holds if, and only if, $n/2- |n/2 - \ka_a + 1| > \delta$, $a\in 1\col 2$, and $n/2- |n/2 - \ka_2 + \ka_1- 1| > \delta$. The first of those two assertions holds if, and only if, $\ka_a\in (\delta+1,n+1-\delta)$ for $a\in 1\col 2$. But since $\ka_1< \ka_2$, the last statement is equivalent to $\delta+1 < \ka_1$ and $\ka_2 < n+1-\delta$, which imply $\ka_2-\ka_1<\ka_2-1-\delta \le n-1-\delta$. The second assertion holds if, and only if, $\delta-1 < \ka_2-\ka_1 < n-1-\delta$, where the second inequality was already implied by the first assertion.
\end{proof}

\begin{proof}[\textbf{Proof of Lemma~\ref{Lem:H1>H0}}]
The proof of this Lemma relies on the findings of Lemma~\ref{Lem:HHc}. Observe first that $\HH^2(\ka_a) > \HH^1$ is equivalent to
\[
	\Delta
	\ := \
	\left| \frac{n}{2} - \eta^2(a) \right|
	-	
	\left| \frac{n}{2} - \eta^1 \right|
	\ > \
	\delta.
\]
Let us denote by $\x^*_{a}$ the path which is equal to $1$ on $1\col (\ka_a-1)$, and equals $2$ on $\ka_a\col n$. Then, if we define
\begin{align*}
	C_1
	\ :&= \
	\|\x - \bone\|_1 \\
	& = \  \|\x_{1:(\ka_a-1)} - \bone\|_1 + (n - \ka_a + 1) - \|\x_{\ka_a:n} - \btwo\|_1, \\
	C_2
	\ :&= \
	\|\x - \x^*_a\|_1
	\ = \
	\|\x_{1:(\ka_a-1)} - \bone\|_1 + \|\x_{\ka_a:n} - \btwo\|_1,
\end{align*}
and note that $C_1$ is equal to $\eta^1$ if $x_1=1$, and to $n-\eta^1$ if $x_1=2$, and likewise that $C_2$ is equal to $\eta^2(a)$ if $x_1=1$, and to $n-\eta^2(a)$ if $x_1=2$. In consequence, we find that
\[
	\Delta
	\ = \
	\left| \frac{n}{2} - C_2 \right|
	-	
	\left| \frac{n}{2} - C_1 \right|.
\]
Next, define
\begin{align*}
    D_1
	\ :&= \
	\frac{1}{2} - \frac{\|\x_{1:(\ka_a-1)} - \bone\|_1}{\ka_a-1} \\
	D_2
	\ :&= \
	\frac{1}{2} - \frac{\|\x_{\ka_a:n} - \btwo\|_1}{n-\ka_a+1},
\end{align*}
and study all four combinations $(\pm D_1,\pm D_2)\ge (0,0)$.

First of all, when $D_1,D_2\ge 0$, we find that
\[
	C_2
	\ \le \
	\frac{1}{2}(\ka_a-1) + \frac{1}{2}(n-\ka_a+1)
	\ = \
	\frac{n}{2}.
\]
Simple calculations then show that $\Delta > \delta$ is equivalent to
\[
	D_1 \ > \ \frac{1}{2} \cdot \frac{\delta}{\ka_a-1}
	\quad
	\text{and}
	\quad
	D_2 \ > \ \frac{1}{2} \cdot \frac{\delta}{n - \ka_a + 1}.
\]
Second of all, when $D_1,D_2\le 0$, we find that
\[
	C_2
	\ \ge \
	\frac{1}{2}(\ka_a-1) + \frac{1}{2}(n-\ka_a+1)
	\ = \
	\frac{n}{2},
\]
so that $\Delta > \delta$ is equivalent to
\[
	D_1 \ < \ -\frac{1}{2} \cdot \frac{\delta}{\ka_a-1}
	\quad
	\text{and}
	\quad
	D_2 \ < \ -\frac{1}{2} \cdot \frac{\delta}{n - \ka_a + 1}.
\]
Third, when $D_1 \le 0 \le D_2$, we have that
\[
	C_1
	\ \ge \
	\frac{1}{2}(\ka_a - 1) + (n-\ka_a+1) - \frac{1}{2} (n - \ka_a + 1)
	\ = \
	\frac{n}{2}.
\]
But in that case, $\Delta > \delta$ is equivalent to
\[
	D_1 \ > \ \frac{1}{2} \cdot \frac{\delta}{\ka_a-1}
	\quad
	\text{or}
	\quad
	D_2 \ < \ -\frac{1}{2} \cdot \frac{\delta}{n - \ka_a + 1},
\]
which is incompatible with the assumption that $D_1 \le 0 \le D_2$. Finally, when $D_2 \le 0 \le D_1$, we have that $C_1 \le n/2$, so $\Delta > \delta$ is equivalent to
\[
	D_1 \ < \ -\frac{1}{2} \cdot \frac{\delta}{\ka_a-1}
	\quad
	\text{or}
	\quad
	D_2 \ > \ \frac{1}{2} \cdot \frac{\delta}{n - \ka_a + 1},
\]
which is again incompatible with $D_2 \le 0 \le D_1$.

In conclusion, $\HH^2(\ka_a)>\HH^1$ is equivalent to either one of the two pairs of inequalities
\[
	\begin{cases}
		\frac{\|\x_{1:(\ka_a-1)} - \bone\|_1}{\ka_a-1}
		\ < \
		\frac{1}{2} - \frac{1}{2} \cdot \frac{\delta}{\ka_a-1}, \\
		\frac{\|\x_{\ka_a:n} - \btwo\|_1}{n - \ka_a + 1}
		\ < \
		\frac{1}{2} - \frac{1}{2} \cdot \frac{\delta}{n - \ka_a + 1},
	\end{cases}
\]
or
\[
	\begin{cases}
		\frac{\|\x_{1:(\ka_a-1)} - \bone\|_1}{\ka_a-1}
		\ > \
		\frac{1}{2} + \frac{1}{2} \cdot \frac{\delta}{\ka_a-1}, \\
		\frac{\|\x_{\ka_a:n} - \btwo\|_1}{n - \ka_a + 1}
		\ > \
		\frac{1}{2} + \frac{1}{2} \cdot \frac{\delta}{n - \ka_a + 1}.
	\end{cases}
\]
Multiplying by $-1$ and adding $1$ to the latter set of inequalities, and simplifying the left-hand sides yield the desired result.
\end{proof}

The arguments of the proof of Lemma~\ref{Lem:H2>H0} are essentially the same as those given in the proof of Lemma~\ref{Lem:H1>H0}. The proof is therefore omitted.

{
\section{Number of segments in a Markov chain}

\begin{Lemma}
\label[Lemma]{Lem:ExpNumSeg}
Let $\X:=(X_k)_{k=1}^n$ be a homogeneous Markov chain with state space $1\col m$, initial distribution $\bpi$ and transition matrix $\p\in [0,1]^{m\times m}$. Then, the expected number of segments of $\X$ is equal to 
\begin{equation*}
n - \bpi \left( \sum_{k = 0}^{n-2} \p^{k}\right) \diag(\p),
\end{equation*}
where $\diag(\p)$ denotes a vector in $[0,1]^{m}$ consisting of the diagonal entries of $\p$. As a consequence, we have the following special cases:

(i) If $\bpi$ is an invariant distribution, then the expected number of segments of $\X$ is 
equal to $n - (n-1)\bpi \diag(\p)$.

(ii) If $\bpi = m^{-1} \bone$ and  $\p$ is equal to $1-p(n,s)$ on the diagonal and $(m-1)^{-1}p(n,s)$ elsewhere, with exit probability $p(n,s)=(s-1)/(n-1)$, then the expected number of segments of $\X$ is equal to $s$.
\end{Lemma}

\begin{proof}
For each $k \in 1\col(n-1)$, it holds that
\begin{align*}
 \Pr(X_{k} \neq X_{k+1}) & = 1 -  \Pr(X_{k} = X_{k+1})  \\
 & = 1 - \sum_{i = 1}^m  \Pr( X_{k+1} = i \mid X_{k} =i)\Pr(X_k = i) \\
 & = 1 - \bpi \p^{k-1}\diag(\p).
\end{align*}
Thus, the expected number of segments of $\X$ is equal to
\[   1 + \Ex\left(
     \sum_{k=1}^{n-1} 1_{[X_{k} \neq X_{k+1}]}
    \right)
     = 
    1 + \sum_{k=1}^{n-1} \Pr(X_{k} \neq X_{k+1}) 
     = 
    n - \bpi \left( \sum_{k = 0}^{n-2} \p^{k}\right) \diag(\p).
\]

\textbf{Part~(i).} It follows from the fact that $\bpi \p^k = \bpi$. 

\textbf{Part~(ii).} Note that  $\bpi$ is the invariant distribution. The assertion follows from part~(i) together with $\bpi\diag(\p) = (n-s)/(n-1)$.
\end{proof}

\begin{Remark}
In practice, the Markov chain $\X$ is often reversible (i.e.\ satisfying the detailed balance condition), where the transition matrix $\p$ is  diagonalizable, and can be written as $\p = \boldsymbol{V} \boldsymbol{\Lambda}\boldsymbol{V}^{-1}$ with $\boldsymbol V\in\R^{m\times m}$ an invertible matrix consisting of eigenvectors, and $\boldsymbol{\Lambda} \in [-1,1]^{m\times m}$ a diagonal matrix of eigenvalues $\{\lambda_i : i \in 1:m\}$. Then, we can simplify the computation by 
the relation
\[
\sum_{k = 0}^{n-2}\p^k =    \boldsymbol{V} \left(\sum_{k = 0}^{n-2}\boldsymbol{\Lambda}^k\right)\boldsymbol{V}^{-1},
\]
where $\sum_{k = 0}^{n-2}\boldsymbol{\Lambda}^k$ is a diagonal matrix with its $i$-th diagonal entry being
$$
\sum_{k = 0}^{n-2}\lambda_i^k = \begin{cases}
n-1 & \text{ if }\lambda_i = 1,\\
\frac{1-\lambda^{n-1}}{1-\lambda_i} & \text{ otherwise. }
\end{cases}
$$
\end{Remark}
}

{
\section{Additional simulations in a misspecified setting}\label{sec:miss}
To examine how deviations from the assumed transition matrix influence QATS, we repeated the simulation study of Section~4 of the main text under controlled perturbations of the matrix $\p$. For this, we fixed the sequence length to $n=10^5+1$, the state space to $m=2$, and the emission distribution to be normal with standard deviation $\sigma=1.0$. For the expected number of segments~$s$, the usual sequence $\{1,2,5,\ldots,2000\}$ is used, influencing the exit probability $p=(s-1)10^{-5}$ and therefore the transition matrix $\p$ with $p_{11}=p_{22}=1-p$ and $p_{12}=p_{21}=p$, which will later be perturbed. These parameters act as the underlying truth and are therefore used for data generation.

To perturb $\p$, we initially leave $p_{11}$ and $p_{22}$ unchanged, but multiply $p_{ij}$ for $i\neq j$ by independents random draws from a uniform distribution on $[\nu^{-1}, \nu]$, where $\nu\in\{1,2,5,10,15,20\}$. Because this random scaling destroys stochasticity, the rows were renormalised so that they again sum to one, resulting in the transition matrix $\tilde\p$. The parameter $\nu$ tunes the severity of the perturbation: When $\nu=1$, this reproduces the true matrix $\p$, whereas larger values of $\nu$ allow off–diagonal probabilities to be multiplied by as much as $\nu$ or shrunk by its reciprocal. After renormalisation the resulting exit probability can differ by an order of magnitude; for example, with $\nu=20$ the chance of switching states is typically ten times higher than in the data‑generating model. This can be observed in the first two rows of Figure~\ref{Fig:Misspec} which display the median and interquartile range of the distribution of the relative difference $100 \cdot (\tilde p_{1j}-p_{1j})/p_{1j}$, for $j=1,2$.

\begin{figure}[!t]
\centering
\includegraphics[width=1\textwidth]{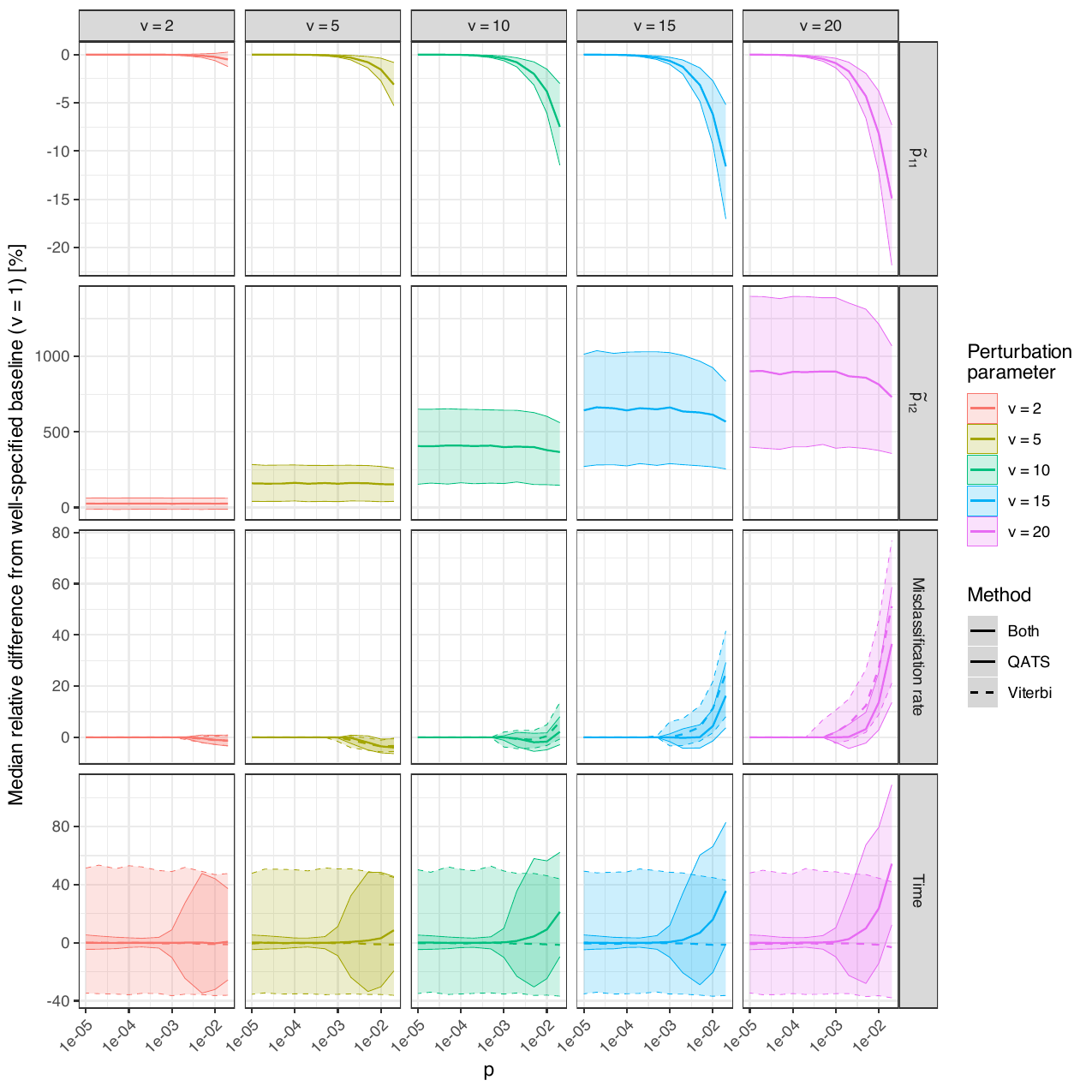}
\caption{Results of Monte Carlo simulations to assess the effect of transition matrix misspecification in the setting $n=10^5+1$, $m=2$, and $\sigma = 1$.} 
\label{Fig:Misspec}
\end{figure}

For every replicate we measured computation time and segmentation error under each perturbation level and expressed the result as a relative difference from the well‑specified baseline ($\nu=1$). Thus, for QATS we define $\Delta T^{\text{Q}}(\nu):=100\cdot(T^{\text{Q}}(\nu) - T^{\text{Q}}(1))/T^{\text{Q}}(1)$. Analogous contrasts were computed for Viterbi and the misclassification rate $d_0$. Repeating this $10^4$ times per setting and summarizing the resulting $\Delta$-values by their median and interquartile range yields a clear picture of how runtime and accuracy change as the transition matrix is increasingly distorted.

Figure~\ref{Fig:Misspec} summarizes the findings. For mild perturbations ($\nu\le 5$), both decoders show a lower median misclassification rate. The relative difference in $d_0$ is negative because the inflated off‑diagonal entries encourage additional state switches and thus better align the estimated path with the truth. Viterbi enjoys this gain at no extra cost, whereas QATS pays a modest runtime penalty. For severe perturbations ($\nu \ge 15$), error rates rise for both methods, yet the increase is consistently smaller for QATS, indicating greater robustness. At $\nu = 20$ and $p=10^{-2}$, the median excess error of Viterbi is roughly one‑and‑a‑half times that of QATS. The additional segments detected by QATS inflate its runtime, especially when the true exit probability $p$ is already sizable, but the overhead remains moderate.
}

\end{document}